\documentclass[twocolumn,10pt]{IEEEtran}
\usepackage{bm,cite,algorithm,float,amsmath,amssymb}

\usepackage{graphicx,epstopdf}
\usepackage{pifont}
\usepackage{cite}
\usepackage{subfigure}
\usepackage{dsfont}

\usepackage{bbm}

\makeatletter

\newcommand{\Rmnum}[1]{\expandafter\@slowromancap\romannumeral #1@}
\makeatother
\usepackage{mathrsfs}
\usepackage{amsfonts}

\usepackage{textcomp}
\usepackage{color}

\usepackage{epsfig}

\usepackage[noend]{algpseudocode}
\IEEEoverridecommandlockouts

\usepackage{mathtools}
\usepackage{amsthm}

\newtheorem{theorem}{Theorem}

\newtheorem{definition}{Definition}

\newtheorem{remark}{Remark}

\newtheorem{lemma}{Lemma}

\begin{document}
		
\title{Analyzing Grant-Free Access for URLLC Service}

\author{
	\IEEEauthorblockN
{ Yan Liu,~\IEEEmembership{Student Member,~IEEE,}
		Yansha Deng,
		~\IEEEmembership{Member,~IEEE,}
		 Maged Elkashlan,~\IEEEmembership{Member,~IEEE,}\\
		 Arumugam Nallanathan,~\IEEEmembership{Fellow,~IEEE,}
		 and George K. Karagiannidis,~\IEEEmembership{Fellow,~IEEE}
}\\

\vspace{-0.2cm}
\thanks{ Manuscript received Feb. 1, 2020; revised Jun. 7, 2020; accepted Jul. 17, 2020. }
\thanks{Y. Liu,  M. Elkashlan and A. Nallanathan are with School of Electronic Engineering and Computer Science, Queen Mary University of London, London E1 4NS, U.K.
%N. Jiang, Y. Deng and A. Nallanathan are with Department of Informatics, King's College London, London, UK 
(e-mail:\{yan.liu, maged.elkashlan, a.nallanathan\}@qmul.ac.uk). }
\thanks{Y. Deng is with Department of Engineering, King's College London, London WC2R 2LS, U.K. 
(Corresponding author: Yansha Deng (e-mail:yansha.deng@kcl.ac.uk).)}
% \thanks{J. Yuan is with School of Electrical Engineering and Telecommunications, The University of New South Wales, Sydney, Australia
% (e-mail:j.yuan@unsw.edu.au).
% %(Corresponding author: Yansha Deng (e-mail:yansha.deng@kcl.ac.uk).)
% }
\thanks{G. K. Karagiannidis is with Department of Electrical and Computer
Engineering, Aristotle University of Thessaloniki, Thessaloniki 54 124, Greece
(e-mail: geokarag@auth.gr).}

%\vspace*{+0.3cm}
}

\maketitle

%\vspace*{-2cm}         
\begin{abstract}
% , such as
% factory automation, autonomous driving, industrial control, and robotic surgeries.
% self-driving cars, industrial control, and real-time gaming. 
%Such services
5G New Radio (NR) is expected to support new ultra-reliable low-latency communication (URLLC) service targeting at supporting the small packets transmissions with very stringent latency and reliability  requirements.  
Current Long Term Evolution (LTE) system has been designed based on grant-based (GB) (i.e., dynamic grant) random access, which can hardly support the URLLC requirements.
%is far from supporting such services, especially for uplink.
Grant-free (GF) (i.e., configured grant) access is proposed as a feasible and promising technology to meet such requirements, especially for uplink transmissions, which effectively saves the time of requesting/waiting for a grant.
% However, this comes at an increased likelihood of collisions resulting from uncontrolled channel access, when the same resources are preallocated to a group of users. 
While some basic GF access features have been proposed and standardized in NR Release-15, there is still much space to improve.
Being proposed as 3GPP study items, three GF access schemes with Hybrid Automatic Repeat reQuest (HARQ) retransmissions including Reactive,  K-repetition, and  Proactive, are analyzed in this paper. 
Specifically, we present a spatio-temporal analytical framework for the contention-based GF access analysis. 
Based on this framework, we define the latent access failure probability to characterize URLLC reliability and latency performances. 
We propose a tractable approach to derive and analyze the latent access failure probability of the typical UE under three GF HARQ schemes.
% Finally, based on the derived latent access failure probability , we derived the latent energy consumption in each TTI of the typical UE for the four schemes, respectively.
Our results show that under shorter latency constraints, the Proactive scheme provides the lowest latent access failure probability, whereas, under longer latency constraints, the K-repetition scheme achieves the lowest latent access failure probability, which depends on $K$. If $K$ is overestimated, the Proactive scheme provides lower latent access failure probability than the K-repetition scheme.
% K-repetition scheme could not offer good latent access failure probability performance in the lower SINR threshold and high UE density scenario, but could provide better latent access failure probability in other scenarios under longer latency constraints.
% We show that the contention based grant-free transmission with proper HARQ design can well meet the reliability requirement.
% It can be demonstrated from our results that the proposed grant-free transmission schemes are able to work together and result in significant performance gains in terms of latency, reliability as well as energy.
% Our results show that: 1) Power boost scheme almost outperforms other schemes in terms of the latent access failure performance but costs too much higher energy than other schemes; 2) K-repetition scheme achieves a good performance in terms of the reliability and latency-energy tradeoff in higher SINR threshold scenarios, but is not suited for scenarios with a lower SINR threshold; 3) The early termination makes the Proactive scheme achieve a good performance in some latency budget, as it saves time and energy wasted by the blind repetitions.

\end{abstract}

%\vspace*{+0.3cm}

\begin{IEEEkeywords}
URLLC, 5G NR, Grant free  access, HARQ
\end{IEEEkeywords}

\section{Introduction}
%The emerging 5G wireless networks are expected to support diverse use-cases which can be broadly classified into three categories [1] enhanced mobile broadband (eMBB), massive machine type communications (mMTC) and ultra-reliable low-latency communications (uRLLC).
%Owing to stringent reliability and latency targets, the most challenging design requirements are created by uRLLC which is the key enabler for the various critical applications across different vertical industries [2].
The Fifth Generation (5G) New Radio (NR) considers three new communication service categories: enhanced Mobile Broadband (eMBB), massive Machine-Type Communications (mMTC), and Ultra-Reliable Low-Latency Communications (URLLC) \cite{3gpp2018study}\cite{series2015imt}.
Among them, the URLLC service is an essential element for applications, including factory automation\cite{7497764}, automation vehicles\cite{8246845}, remote control\cite{8663990}, and virtual/augmented reality (VR/AR)\cite{8663985}, which has stringent requirements on low latency and high reliability for small packets transmissions.
% (e.g., generally $1-10^{-5}$ reliability with 1ms user plane latency, and 0.5ms for both downlink (DL) and uplink (UL)\cite{3gpp2018study}) 
% A general URLLC requirement is $1-10^{-5}$ reliability with 1ms user plane latency for 32 bytes (0.5ms for both downlink (DL) and uplink (UL))\cite{3gpp2018study}.
The Third Generation Partnership Project  (3GPP) has defined a general URLLC requirement: $1-10^{-5}$ reliability within 1ms user plane latency\footnote{User plane latency is defined as the one-way latency from the processing of the packet at the transmitter to when the packet has been received successfully and includes the transmission processing time, transmission time and reception processing time.} 
for 32 bytes (0.5ms for both downlink (DL)  and uplink (UL))\cite{3gpp2018study}.
% A general URLLC reliability requirement for one transmission of a packet is 1-10-5 for 32 bytes with a user plane
% latency of 1ms.
More details about the variety of different traffic characteristics and the requirements of some URLLC use cases can be found in \cite{3gpp22261}. %\cite{kim2016grant}\cite{3gpp22261}.
% \cite{8474959}
For example, the automation use case requires $1-10^{-5}$ reliability within  10ms for remote motion control; 
the intelligent transportation use case requires $1-10^{-6}$ reliability within 5ms
for cooperative collision avoidance.

Current Long Term Evolution (LTE) system can hardly fulfill  the URLLC requirements. Especially in the uplink, current LTE utilizes a scheduling based transmission mode, namely, grant-based (GB) scheduling as specified in \cite{36213}.
% a User Equipment (UE) needs to send scheduling request to
% the serving base station in a dedicated and periodic resource
% and then wait for the scheduling grant from the base station
% in some slots later.
This conventional GB scheduling is initiated by the User Equipment (UE) with an access request to the network in which the Base Station (BS) can respond by issuing an access grant through a four-step random access (RA) procedure as shown in Fig. 1.
Such scheduling-request-triggered transmission would take at least 10ms before starting the data transmission, which is far from the URLLC latency requirement.
% traffic is transmitted within 0.5 ms after the uplink traffic arrived at the buffer for transmission 
Recently, grant-free (GF) access has been proposed and extensively discussed in 3GPP RAN WG1 \cite{1167309,11705245,11705654} to cope with the URLLC requirement in the uplink transmission. 
% This scheduling-request-triggered transmission is not fast enough to guarantee that the URLLC traffic is transmitted within 0.5 ms after the uplink traffic arrived at the buffer for transmission.
With uplink GF access, a UE with a small packet can transmit data along with required control information in the first step transmission itself. 
This can greatly reduce the RA and data transmission latency, as the scheduling request and grant issuing step in  GB RA are removed as shown in Fig. 1.

In the GF transmission, the frequency resource can be reserved in advance or allocated at the time when there is a request. Preallocation of the dedicated resource, known as Semi-Persistent-Scheduling (SPS)\cite{1167309}, is more suitable for periodic traffic with a fixed pattern,  whereas contention-based GF transmission over the shared resource is more suitable for sporadic packets, as it is more efficient and flexible in terms of resource utilization.
However, contention-based GF transmission is subject to potential collisions with other neighbouring UEs transmitting simultaneously over the shared resource, thus jeopardizing the transmission reliability.

A standard technique to improve transmission reliability, which has been adopted in various wireless standards, is Hybrid Automatic Repeat reQuest (HARQ) retransmission\cite{vangelista2018performance}.
Conventional HARQ allows for retransmissions only upon reception of a Negative ACKnowledgement (NACK).
This requires the BS to first receive the packet for detection, then issue the feedback. 
This is the so called \textit{Reactive} (Reac) scheme, where retransmissions are triggered only when there is a failure in the previous transmission.
% To improve reliability, retransmissions schemes such as Hybrid Automatic Repeat reQuest
% (HARQ) have been considered.
%and other HARQ strategies to further reduce latency and improve reliability are needed for URLLC.
%Then, a next level benefit by GF is to have longer time for HARQ retransmissions, and thus implies higher reliability within the same latency budget, compared to GB transmissions.
%However, the reactive HARQ scheme can only support a limited number of retransmissions before the URLLC requirementsis no longer met. 
%Therefore different HARQ strategies to further reduce latency and improve reliability have beenrecently studied. 
%One technique that has been considered for 5G, is to run a number of blind transmissions of the same payload.  The BS can then perform soft combining of the transmissions to improve the decoding reliability [13]. 
%Such kind of solution is already part of the recent 3GPP agreements for NR and are referred to as K-repetition (K-Rep) [14].
% One technique that has been considered for
% 5G, is to run a number of blind transmissions of the same
% payload. 
% The BS can then perform soft combining of the
% transmissions to improve the decoding reliability [13]. Such
% kind of solution is already part of the recent 3GPP agreements
% for NR and are referred to as K-repetition (K-Rep) [14
However, the Reactive scheme introduces additional latency, as the UE needs to wait for the feedback before performing a retransmission, which is determined by the HARQ round-trip-time (RTT), i.e., the time duration of the cycle from the beginning of the transmission until processing its feedback\cite{sesia2011lte}.
Thus, the Reactive scheme only allows for a limited number of retransmissions due to the stringent latency requirement of URLLC service\cite{sesia2011lte}, and this fact motivates more research for advanced HARQ schemes to be integrated with GF transmission to provide reduced latency and enhanced reliability.
% A standard technique, which has been adopted in various wireless standards, Hybrid Automatic Repeat reQuest (HARQ) retransmission, can be integrated with GF transmission 
% Then, GF transmissions with Hybrid Automatic Repeat reQuest (HARQ) retransmissions can be adopted 
% to improve the reliability.
% within the same delay budget, compared to GB transmissions.without scheduling request and dynamic grant, 
% However, the benefits of time diversity could be rather limited under stringent latency constraints for URLLC service as the number of HARQ round trips and channel uses is rather limited. 

% In this scenario, new HARQ schemes can be integrated with GF transmission.

% However, the benefit of feedback-based retransmissions (even with error-free but delayed feedback) is questionable since each transmit packet is much smaller due to energy and latency constraints, thus more prone to errors. 

\begin{figure}
	\centering
	\includegraphics[width=3.5in,height=2.5in]{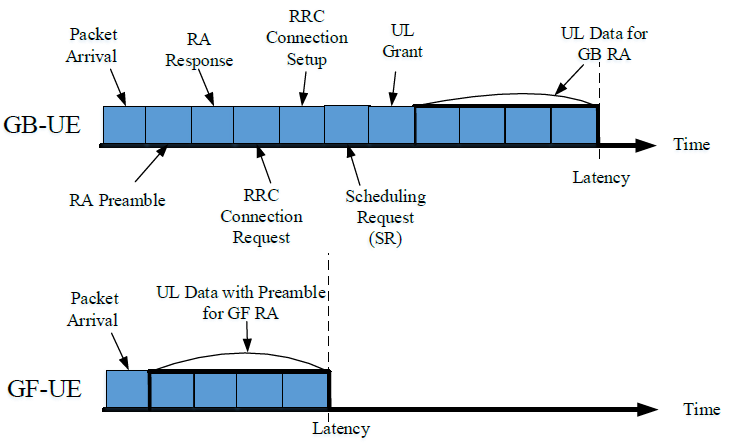}
	\caption{Uplink transmissions for grant-based and grant-free random access}
	\label{fig:my_label}
\end{figure}

% The potential collision however necessitates different means of collision avoidance and resolution in order to meet reliability targets.

One scheme is the \textit{K-repetition} (Krep) scheme supported  in  the 3GPP NR Release-15\cite{3gpp38214}, where the pre-defined number ($K_{\rm Krep}$) of consecutive replicas of the same packet are transmitted without waiting for the feedback, and then the BS performs soft combining of these repetitions to improve the reliability \cite{1705246}.
Another candidate scheme is known as the \textit{Proactive} (Proa) scheme, which has been discussed  in\cite{1612246}\cite{1903079}.
In a Proactive scheme, the UE still repeats transmissions in consecutive transmission time intervals (TTIs) like K-repetition scheme with maximum $K_{\rm Proa}$ times, but if the UE receives and decodes a positive feedback (ACK) from the BS before reaching maximum $K_{\rm Proa}$ times, the repetition will be terminated to reduce latency. 
It is noted that this scheme is more computational heavy for the UE, as the UE has to monitor the feedback.
% However, it is also likely to be more resource efficient than K-repetition scheme if the number of blind repetitions is overestimated and more reliable if the number is underestimated.

Another standard technique to enhance reliability is the efficient random access control mechanism, including the Access Class Barring (ACB), the Back-Off (BO) and the Power Boosting (PB) schemes\cite{nan2018random}.
However, both the ACB and  BO schemes make a group of UEs completely barred in specific time slots, which will introduce extra latency for these UEs.
As such, we just consider the GF HARQ schemes integrated with the PB, which can quickly compensate unexpected Signal-to-Interference plus-Noise Ratio (SINR) degradations at the initial transmissions\cite{4Abreu}.
% This is motivated by the fact that HARQ retransmissions within the latency deadline should be prioritized to improve the success probability \cite{4Abreu}. 
Specifically, if a transmission fails, the UE uses the full path-loss inversion power control to maintain the average received power at a higher power level in the next retransmission, where the power control is one candidate technology component for uplink transmission with the focus on improving the reliability.

Despite that the aforementioned GF access designs are proposed to govern the URLLC service, their theoretical formulations and comparative insights have never been fully established.
Recent works \cite{4Abreu}\cite{8877253} have evaluated the Reactive, K-repetition, and Proactive schemes for URLLC service through system-level simulation without analytical characterization.
The authors in \cite{4Abreu} claimed that the effects of inter- and intra-cell interference, queuing and time-frequency variant channels, are difficult or even infeasible to evaluate with analytical models.
% The authors in \cite{4Abreu} have evaluated the effect of power control and power boosting for GF  transmissions using system level simulations.
% These simulations permit to study effects of inter- and intra-cell interference, queuing and time-frequency variant channels, which would be difficult or even unfeasible to evaluate with analytical models.
This is because existing wireless systems were designed mainly to maximize the data rates of the long packet transmission,  the short packet transmission in URLLC service challenges the existing wireless system in terms of the joint reliability and latency requirements.
To cope with it, correctly modeling and analyzing the reliability and latency is fundamentally important, but the interplay between latency and reliability brings extra complexity.
In this paper, we address the following fundamental questions:
1) how to quantify the URLLC reliability and latency performances;
2) how to examine whether  different GF schemes satisfy the URLLC reliability and latency performances or not;
3) how to evaluate which GF scheme performs better in a certain specific scenario.
To do so, we present a novel spatio-temporal mathematical framework to analyze and evaluate both the reliability and latency performances for three different GF HARQ schemes. 
The main contributions of this paper can be summarized in the following:
% \begin{itemize}
%     \item We give an overview of contention-based GF random access schemes in 5G NR and discussing their procedures.
% \end{itemize}

\begin{itemize}
    \item We present a novel spatio-temporal mathematical framework for analyzing contention-based GF HARQ schemes for URLLC service  by using stochastic geometry and probability theory.
    In the spatial domain, stochastic geometry is applied to model and analyze the  mutual interference among active UEs (i.e., those with non-empty data buffer). In the time domain, probability theory is applied to model the correlation of  the buffer state and the transmission state over different time slots.
\end{itemize}
% \begin{itemize}
%     \item In the spatial domain, stochastic geometry is used to account for mutual interference among active UEs (i.e., those with non-empty data buffer). In the time domain, probability theory is used to account for the buffer state and the transmission state.
% \end{itemize}

\begin{itemize}
    \item Based on this framework, we propose a tractable approach to characterize and analyze the URLLC performances of a randomly chosen UE by defining the latent access failure probability.
    We then derive the exact closed-form expressions for the latent access failure probabilities of the UE under three different contention-based GF HARQ schemes, including Reactive, K-repetition, and Proactive schemes, respectively.
\end{itemize}

% \begin{itemize}
%     \item 
% \end{itemize}

% \begin{itemize}
%     \item In terms of the latent access failure probability , we first derive the exact expressions for four different GF schemes without HARQ retransmissions, respectively, and then extend the analysis to GF schemes with HARQ retransmissions. Finally based on the derived latent access failure probability , we derive the expressions of the energy consumption in each TTI of the UE for the four GF HARQ schemes, respectively.
% \end{itemize}

% \begin{itemize}
%     \item 
%     We first derive the exact expressions of the latent access failure probability for different GF schemes without HARQ retransmissions, respectively, and then extend the analysis to GF schemes with HARQ retransmissions. We also derive the expressions of the latency and the energy consumption of the UE for different GF schemes.
% \end{itemize}

\begin{itemize}
    \item We develop a realistic simulation framework to capture the randomness locations, pilot and data transmissions as well as the real packets of each UE in each TTI to verify our derived latent access failure probability. 
    We compare  the effectiveness of the three different GF HARQ schemes.
    Our results show that the Proactive scheme provides the lowest latent access failure probability under shorter latency constraints, while the K-repetition scheme has the lowest latent access failure probability as well as the most improvement with PB under longer latency constraints.
    % We believe the presented model and results can be used as a reference for empirical system-level analysis of uplink grant free solutions.
\end{itemize}

% \begin{itemize}
%     \item The analytical model presented in this paper can also be applied for the reliability and latency performance evaluation of other types
% of GF HARQ schemes in the cellular-based networks.
% \end{itemize}

The rest of the paper is organized as follows. Section II provides the problem formulation and system model. 
Section III analyzes the URLLC performance by deriving the expressions of the latent access failure probability of a randomly chosen UE under three different GF HARQ schemes.
Section  IV  provides numerical results. 
Finally, Section V  concludes the work.

% 

%Techniques to improve the supported load with GF random access while ensuring high reliability and low latency is are therefore currently being
%discussed in academic rese\cite{181477}. 
%State of the art solutions include GF transmissions with K-repetition, where a pre-defined number of replicas are transmitted, and proactive repetition\cite{8340053}.

\section{Problem Formulation and System Model}

\subsection{Network Model}
We consider a single layer cellular network, where the BSs and the UEs are spatially distributed following two independent Poisson Point Processes (PPPs) $\Phi_{\rm{B}}$ and $\Phi_{\rm{D}}$ with intensities $\lambda_{\rm{B}}$ and $\lambda_{\rm{D}}$, respectively. 
% The deployed BSs provides coverage to the URLLC UEs.
We assume that each UE associates to its geographically nearest BS, where a Voronoi tessellation is formed. 
The UEs are connected and synchronized to the serving cell. 
% The semi-static configuration includes time and frequency resource allocation,
% modulation and coding scheme (MCS), power control settings and HARQ related parameters
Moreover, we consider additive noise with average power $\sigma ^2$ and a flat Rayleigh fading channel, i.e. the channel response is constant over the selected Resource Blocks (RBs), however, it can vary at every transmission or retransmission. 
The channel power gain $h$ is assumed to be exponentially distributed  with unit mean, i.e., $h\sim$ Exp(1).
All channel gains are assumed to be independent and identically distributed (i.i.d.) in space and time.
We consider the path loss model with the path-loss attenuation $x^{-\alpha}$, where $x$ is the propagation distance and  ${\alpha}$ is the path-loss exponent. 
We apply a full path-loss inversion power control at all UEs to solve the ``near-far'' problem, where each UE compensates for its own path-loss to keep the average received signal power equal to a same threshold $\rho$.
We also assume the density of BSs is high enough and no UE suffers from truncation outage \cite{nan2018random}.

%such that their transmissions are received at the same average power. %though their instantaneous receive power may change at each transmission due to Rayleigh fluctuations.

 \subsection{Contention-Based Grant-Free Access}

In this paper, we consider the uplink contention-based GF access for UEs with sporadic small packets with URLLC requirements, where UEs transmit data in an arrive-and-go manner without sending a scheduling request and receiving resource grant from the network. 
%According to \cite{nan2018collision}\cite{8533378}, each UE transmits its pilot and data simultaneously.
Each UE has a data buffer that stores packets received from higher layers.
%As only active UEs will try to request for uplink channel resources, we define the active probability of each IoT device pa ∈[0, 1] follows a Bernoulli process.
An i.i.d. Bernoulli traffic generation model with probability of $p_a\in [0, 1]$, is assumed at each buffer.
Note that we only consider a single packet sequence arrival.
This packet sequence will be removed from the buffer, i.e., the buffer becomes empty without new packets,  once it has been successfully transmitted, otherwise, this UE will wait and reattempt in the next HARQ retransmission.
%At each time slot, transmitters with non-empty buffers employ the F-ALOHA protocol with probability pfa  pa/Nc to access one of the Nc channels, where pa is the ALOHA transmission probability.

GF uplink transmissions occur in a slotted-ALOHA system based on OFDM (Orthogonal Frequency Division Multiplexing) within short-TTI\footnote{5G NR introduces the concept of `mini-slots' and supports a scalable numerology allowing the sub-carrier spacing (SCS) to be expanded up to 240 kHz. 
In contrast with the LTE slot duration of
14 OFDM symbols per TTI, mini-slots in 5G NR can
compose of 1-13 symbols.  
Collectively, this allows shorter transmission slots to meet the stringent latency requirement. 
}%Thisallows shorter transmission slots without increasing the SCS, which is particularly suitable for low frequency bands.
In this paper, the TTI refers to a mini-slot, which is shorter than the typical coherence times that are of the order of few milliseconds. But generally, the coherence time could be normalized. In addition, the repetitions could be performed over different RBs in frequency so that the channel gains i.i.d assumption is justified. \cite{7432148}.
%The transmissions occur in a frame based system like LTE and occur in transmission time intervals (TTI). %of mini-slots with 2 OFDM symbols.
%In this work, we assume a four symbol mini-slot at 30 kHz SCS, resulting in a transmission duration of 0.143 ms. 
%These assumptions follow the 3GPP NR URLLC evaluation agreements. 
The UEs are configured by radio resource control (RRC) signaling prior to the GF access (as Type 1 UL \cite{3Chairman}), with time and frequency resource, modulation and coding scheme (MCS), power control settings, and HARQ related parameters\cite{3gpp38824}.
% for GF transmission over a set of preallocated radio resources.
The configured UEs are connected and synchronized, thus being always ready for a GF transmission.
According to \cite{8533378}, we consider $N$ UEs pre-configured with $S$ orthogonal pilots, i.e.,
$S$ sub-carriers over one TTI, for their uplink GF transmissions in the frequency domain\cite{8069007}.
At the beginning of each round trip, UEs randomly move to new positions, and the active ones randomly select one of the available $S$ pilots to transmit with the data simultaneously\cite{nan2018collision}.
A collision occurs when the same pilot is transmitted by two or more UEs using the same sub-carrier, and received successfully by the same BS.
% Note that only active UEs in a TTI cause interference to each other.
%Note that in the 1st transmission, 
% The active factor is the probability of a
% user being active in a TTI.
% The \textit{active factor} of the UE, i.e., the probability that the UE has non-empty data buffer, in the 1st TTI, is ${\mathcal A}_1=p_a$.
According to the thinning process \cite{kingman1993poisson}, the density of active UEs choosing the same pilot can be derived as
\begin{align}\label{density}
{\lambda _{a}} = {p_a\lambda _D}/S.
\end{align}
% Using the RB consists of 7 OFDM symbols in the time domain, and the 60 kHz subcarrier spacing (SCS). The length of each TTI is 0.125 ms, which is much smaller than that in LTE, to allow the possibility to meet the stringent latency requirement.
% These assumptions follow the 3GPP NR URLLC evaluation
% agreements \cite{2Tel2018}.
% This leaves sufficient time budget for the first transmission, processing at the UE and HARQ retransmissions  within the one ms latency target for URLLC services.
%This waiting time is denoted as frame alignment (A). If the packet is successfully decoded, the BS sends an ACK feedback (F), otherwise it sends a NACK. After having received and decoded the feedback, the UE can decide to perform a retransmission (T).

% \begin{figure}
% 	\centering
% 	\includegraphics[width=4in,height=2.5in]{grant-free.PNG}
% 	\caption{Grant-free transmission procedure}
% 	\label{fig:my_label}
% \end{figure}

\subsection{Grant-Free Access Schemes}
% \begin{figure}
% \centering
%  \subfigure[Reactive GF transmission / Power boost GF transmission]
%  {\begin{minipage}[t]{0.50\textwidth}
%  \includegraphics[width=3.8in,height=2.8in]{reactive.pdf}
% % \subcaption{1}
%  \end{minipage}}
%  \subfigure[K-repetition GF transmission with K=4 repetitions]
%  {\begin{minipage}[t]{0.49\textwidth}
%  \centering
%  \includegraphics[width=3.8in,height=2.8in]{krepe1.pdf}
% % \subcaption{1}
%  \end{minipage}}
%  \subfigure[Proactive GF transmission with maximum K=8 repetitions ]
%  {\begin{minipage}[t]{1\textwidth}
%  \centering
%  \includegraphics[width=6in,height=3.6in]{proactive.pdf}
% % \subcaption{1}
%  \end{minipage}}
% \caption{The Uplink Grant-Free HARQ Schemes for URLLC. A = Frame alignment, T = Transmission, DP = Downlink Processing, UP = Uplink Processing, F = Feedback.} 
% \label{fig:1}
% \end{figure}
This section provides a general description of the three GF HARQ schemes considered in this paper. 
% Four GF HARQ schemes are considered in this paper.
For ease of description, we first present  definitions for general variables.
As illustrated in Fig. 2, the frame alignment (A) delay is denoted as $T_{\rm{fa}}$, the packet transmission (T) time is denoted as $T_{\rm{tx}}$, and the processing (DP) time at the BS is denoted as $T_{\rm{dp}}$. 
If the packet is successfully decoded, the BS sends an ACK feedback, otherwise it sends a NACK, where the ACK/NACK feedback (F) time is represented by $T_{\rm{fd}}$. 
After having received and processed the feedback, the UE can decide whether to perform a retransmission.
The processing time at the UE is denoted by $T_{\rm{up}}$.
The frame alignment delay $T_{\rm{fa}}$ is a random variable uniformly distributed between zero and one TTI\cite{8902865}. 
Depending on the packet size, channel quality and scheduling strategy, the transmission time $T_{\rm{tx}}$ can vary from one to multiple TTIs. 
Considering the small packets of URLLC
traffic, we assume $T_{\rm{fa}}$ = 1 TTI and $T_{\rm{tx}}$ = 1 TTI in this work same as\cite{8835946}. 
The BS feedback time $T_{\rm{fb}}$ and the BS (UE) processing time $T_{\rm{dp}}$ ($T_{\rm{up}}$) are also assumed to be one TTI.
Then, the latency framework of the three GF HARQ schemes are described as follows. 
\subsubsection{\textbf{Reactive scheme}}
% \begin{figure}
% 	\centering
% 	\includegraphics[width=6in,height=3in]{reactive.pdf}
% 	\caption{Reactive GF transmission / Power boost GF transmission}
% 	\label{fig:my_label}
% \end{figure}
\begin{figure}[htbp!]
    \begin{center}
    \begin{minipage}[t]{0.48\textwidth}
    \centering
        \includegraphics[width=3.1in,height=2in]{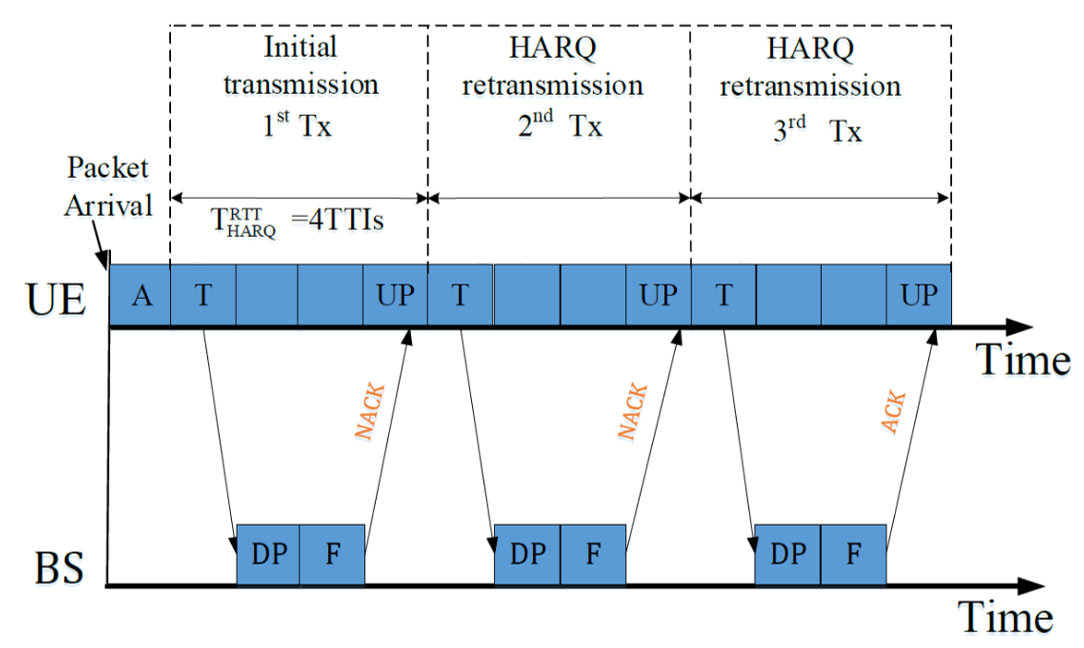}
        \vspace*{-0.4cm}
        \caption{\scriptsize Reactive GF transmission}
    \end{minipage}
    \label{fig:10}
        \begin{minipage}[t]{0.48\textwidth}
    \centering
        \includegraphics[width=3.1in,height=2in]{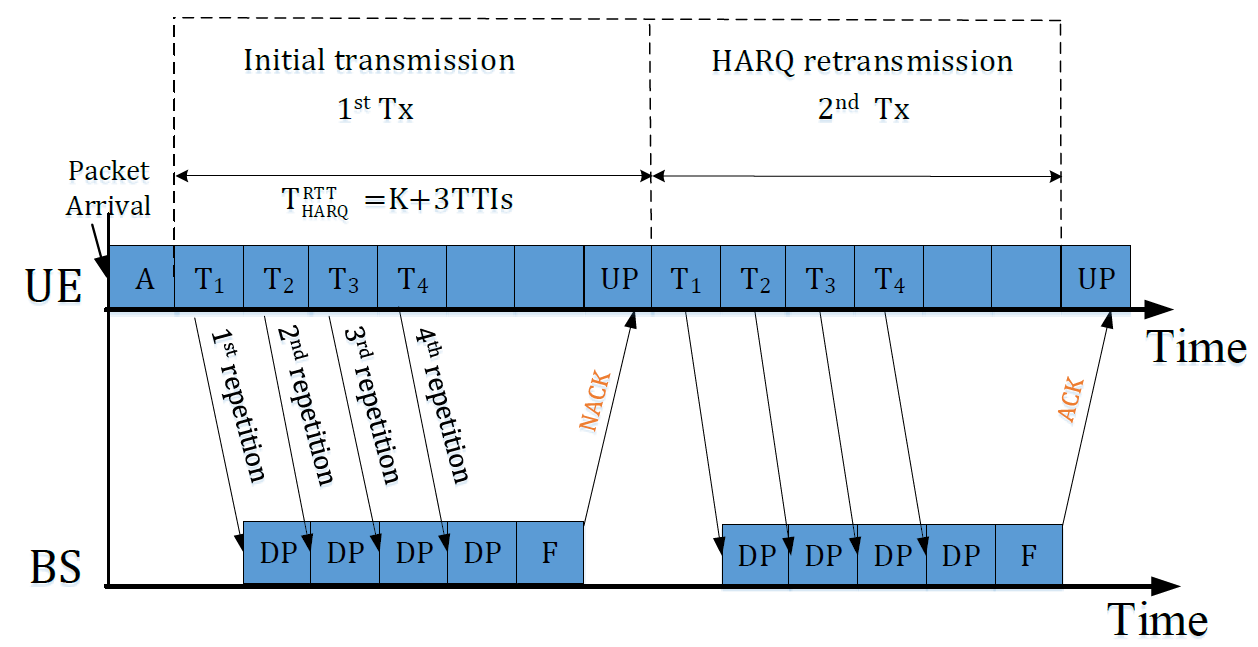}
        \vspace*{-0.4cm}
        \caption{\scriptsize K-repetition GF transmission with $K_{\rm Krep}$=4 repetitions}
        \end{minipage}
        \label{fig:11}
    \end{center}
\end{figure}
% A frame-based system alike LTE is assumed, meaning that transmissions can start on a frame basis. 
%The transmissions occur when the UE is already synchronized and in connected state. 
% The UEs are configured by radio resource control (RRC) signaling (as Type 1 UL \cite{3Chairman}). 
% The needed control information, such as time and frequency resource allocation, modulation and coding scheme (MCS), power control settings and HARQ related parameters, are configured by RRC signaling prior to the GF random access. 
% In the uplink, the configured users are connected and synchronized, thus being always ready for a URLLC transmission. 
%\subsection{Grant Free Schemes}
%GF access with four HARQ schemes are considered in this paper.

%However, transmissions are then susceptible to potential collisions from simultaneous transmissions of neighbouring nodes.
%This section presents a system-level simulations based performance comparison of several uplink 
%GF transmission procedures considering a large urban macro network.
The Reactive scheme is illustrated in Fig. 2. 
After the UE finalizes its initial uplink  transmissions (T), its signal will be processed at
the BS (DP) for a HARQ feedback (F) (ACK/NACK). 
After processing the HARQ feedback (UP), the UE retransmits the same packet upon reception of a NACK. 
In this scheme, we note that the HARQ round trip time 
\begin{align}\label{RTT_reactive}
T^{\rm{RTT}}_{\rm{Reac}}=4 {\rm TTIs}.    
\end{align}
Then the latency after $m$ HARQ round trips is obtained as
\begin{align}\label{latency_rea1}
&{T_{\rm{Reac}}}[m]
 ={T_{\rm{fa}}}+m{T^{\rm{RTT}}_{\rm{Reac}}} \nonumber\\
 &={T_{\rm{fa}}} +m({T_{\rm{tx}}} + 
 {T_{\rm{dp}}}+ {T_{\rm{fb}}} + {T_{\rm{up}}}) \nonumber\\
 &={1+4m} \ {\rm{TTIs}}.
\end{align}

%\begin{align}\label{latency_rea1}
%{T_{\rm{Reac}}}[m]={T_{\rm{fa}}}+m{T^{\rm{RTT}}_{\rm{Reac}}}={T_{\rm{fa}}} +m({T_{\rm{tx}}} + {T_{\rm{dp}}}+ {T_{\rm{fb}}} + {T_{\rm{up}}}).
%\end{align} 

\subsubsection{\textbf{K-repetition scheme}}
% \begin{figure}
% 	\centering
% 	\includegraphics[width=6in,height=3in]{krepe1.pdf}
% 	\caption{K-repetition GF transmission with K=4 repetitions}
% 	\label{fig:my_label}
% \end{figure}
The K-repetition scheme is illustrated in Fig. 3, where the  UE is configured to autonomously transmit the same packet for $K_{\rm Krep}$ repetitions in consecutive TTIs. 
At the end of $K_{\rm Krep}$ repetitions, the BS needs to combine the received repetitions, process the received packet, and feedback to the UE.
In this scheme, the HARQ round trip time 
\begin{align}\label{RTT_krep}
T^{\rm RTT}_{\rm Krep}=(K_{\rm Krep}+3) {\rm TTIs}.    
\end{align}
Then the latency after $m$ HARQ round trips is defined as 
\begin{align}\label{latency_krep}
&T_{\rm{Krep}}[m] 
={T_{\rm{fa}}}+mT^{\rm{RTT}}_{\rm{Krep}} \nonumber\\
&= {T_{\rm{fa}}} + m(K_{\rm Krep}{T_{\rm tx}} + {T_{\rm dp}}+ {T_{\rm fb}} + {T_{\rm up}})
\nonumber\\
&=1+m(K_{\rm Krep}+3)\ \rm{TTIs}.
\end{align}
% Each repetition can be identical, or consists of different redundancy versions of the encoded data. 
% This method can eliminate the RTT latency, with a potential resource wastage if the needed number of repetitions is overestimated.
% This method eliminates the RTT latency at the expense of potential resource wastage if the needed number of repetitions is overestimated.

\subsubsection{\textbf{Proactive scheme}}
The Proactive scheme is illustrated in Fig. 4. Similarly to the K-repetition scheme, the UE is configured to repeat the transmission for a maximum number of $K_{\rm Proa}$ repetitions but can receive the feedback after each repetition. This allows the UE to terminate repetitions earlier once receiving the positive feedback (ACK). 
% For the Proactive scheme with a maximum $K_{\rm Proa}$ repetitions, the description of the latency is more complicated. 
We note that the UE could receive the 1st feedback 3TTIs after the 1st repetition. 
\begin{figure}
	\centering
	\includegraphics[width=3.6in,height=2.1in]{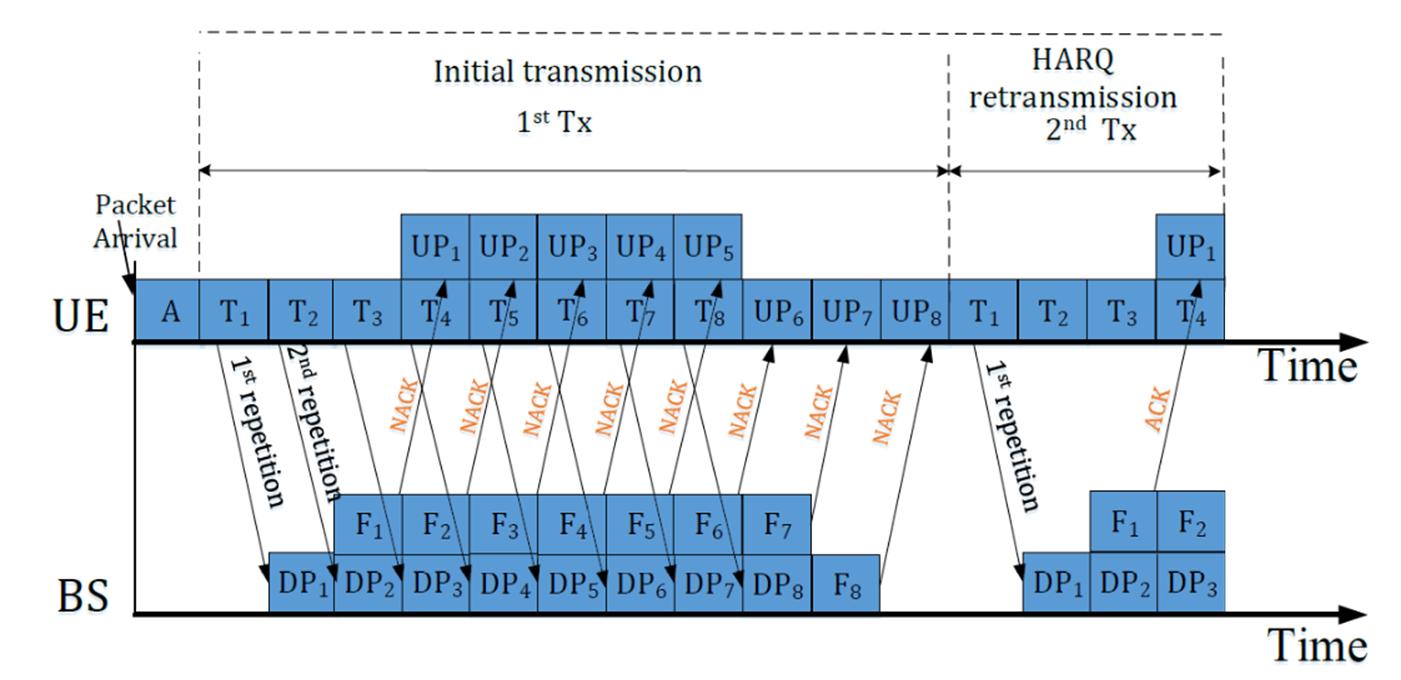}
	\caption{Proactive GF transmission with maximum $K_{\rm Proa}$=8 repetitions}
	\label{fig:my_label}
\end{figure}
That is to say, the minimum HARQ round trip time is 4.
For  $K_{\rm Proa}\leqslant4$ , the UE continues repetitions until maximum $K_{\rm Proa}$.
For $K_{\rm Proa}\geqslant5$, the UE continues repetitions until  either the UE receives ACK from the BS, or the number of repetitions reaches maximum $K_{\rm Proa}$ times\cite{2Tel2018}.
Let us denote the 1st access success of the typical UE occurs in the $l$th repetition during one HARQ round trip.
Thus, we have the single HARQ round trip time for the Proactive scheme as:
 \begin{align}\label{pro_rrt1}
 T_{{\rm Proa},K,l}^{\rm RTT}=
  \begin{cases}
    K_{\rm Proa}+3, 
  &\mbox{     $l=0$},\\
  l+3, 
  &\mbox{   $1\leqslant l\leqslant K_{\rm Proa}$}.
  \end{cases}
\end{align}
% \begin{align}\label{pro_rrt1}
%  T_{{\rm Proa},K,l}^{\rm RTT}=
%   \begin{cases}
%   K_{\rm Proa}+3, 
%   &\mbox{  $K_{\rm Proa} \leqslant 4 $ },\\
%     K_{\rm Proa}+3, 
%   &\mbox{   $K_{\rm Proa} \geqslant  5$  $\&$  $l=0$},\\
%   l+3, 
%   &\mbox{   $K_{\rm Proa} \geqslant  5$ $\&$ $l \le K_{\rm Proa}$}.
%   \end{cases}
% \end{align}
% \begin{align}\label{pro_rrt1}
% %\hspace{-0.5cm}
%  T_{{\rm Proa},K,l}^{\rm RTT}=
% \begin{cases}
% K+3, &\rm{if \ K \leqslant 4 \  or \ if \ K> 4 \ with \ none  \ repetition \ succeeds},\\
% l+3, &\rm{if \   K >4\ and \ the \ first \ success\ occurs \ in \ $the$ \ l $th repetition$}.
% \end{cases}
% \end{align}

Note that if $l=0$ , i.e., all the $K_{\rm Proa}$ repetitions in one HARQ round trip are not successful with $T_{{\rm Proa},K,0}^{\rm RTT}=K_{\rm Proa}+3$, the UE will perform HARQ retransmission in next round trip. 
Thus, the latency after $m$ HARQ round trips for the Proactive scheme with a maximum $K_{\rm Proa}$ repetitions can be derived as
% If all the $K_{\rm Proa}$ repetitions in one HARQ round trip are not successful with $T_{{\rm Proa},K,l}^{\rm RTT}=K_{\rm Proa}+3$, 
\begin{align}\label{latency_pro}
&T_{\rm Proa}[m]={T_{\rm fa}}+\underbrace{(m-1)T^{\rm RTT}_{{\rm Proa},K,0}}_{\rm I}+\underbrace{T^{\rm RTT}_{{\rm Proa},K,l}}_{{\rm I}{\rm I}}\nonumber\\
&={T_{\rm fa}} + (m-1)(K_{\rm Proa}+3)+T^{\rm RTT}_{{\rm Proa},K,l}\nonumber\\
&=l+4+(m-1)(K_{\rm Proa}+3) \ \rm{TTIs} \ (1 \leqslant  \textit{l} \leqslant K_{\rm Proa}),
\end{align}
where I denotes that the transmissions in all the former ($m-1$) HARQ round trips are  not successful; and II implies the possible case in the final $m$th HARQ round trip given in  (\ref{pro_rrt1}).
% A reduction of the overall transmission resources can be obtained compared to the K-repetition scheme in case the time spent for the K times transmission is higher than the HARQ RTT.
% Then a reduction of the overall transmission resources and latency can be obtained compared to the K-repetition scheme, in case K is overestimated; and the enhancement of the reliability can be guaranteed compared to the K-repetition, in case K is underestimated.

\subsection{Signal to Noise plus Interference Ratio (SINR)}
% According to \cite{nan2018collision}\cite{8533378}, each UE chooses a pilot sequence from a predetermined orthogonal pilot sequence set and transmits its selected pilot
% and data simultaneously. 
% A collision occurs when two or more
% devices choose the same pilot sequence.
Note that the GF access failure occurs due to the following two reasons: 1) a pilot cannot be recognized by the received BS, due to its lower received SINR than the SINR threshold $\gamma_{\rm th}$; 
2) the BS successfully receives two or more same pilots simultaneously, such that the collision occurs, and the BS cannot decode any collided pilots.
% random access is successful if 1) the received SINR at the BS is greater than the SINR threshold $\gamma_{th}$; and 2) no collision occurs (i.e., no other UEs successfully transmit to the typical BS simultaneously); otherwise access outage occurs.
Our model follows the assumption of collision model in \cite{nan2018collision}\cite{3gpp2011study}, where all these collision UEs would not be decoded at the BS.
Different from the data transmission with no intra-cell interference due to orthogonal resource allocation, the GF  access analysis in this work needs to take into account both the inter- and intra-cell interference\footnote{We consider intra-cell interference because the UEs in the same cell associated with the same BS may choose the same pilot. 
We consider the inter-cell interference due to that the UEs in different cells share the preamble sequence pool among BSs.
Similar with \cite{7917340}, we focus on providing a general analytical framework of cellular network,  considering both the inter- and intra-interference.}.
% due to that the UEs in the same cell associated to the same BS may choose the same preamble.
We formulate the SINR of a typical BS located at the origin as
\begin{align}\label{SINR1}
{\rm{SINR}}_m = \displaystyle\frac{g_m\rho h_0}{{\mathcal {I}_{\rm intra} +\mathcal {I}_{\rm inter} + {\sigma ^2}}},
\end{align}
where $\rho$ is the full path-loss inversion power control threshold, $h_0$ is the channel power gain from the typical UE to
its associated BS, $\sigma_2$ is the noise power, ${\mathcal {I}}_{\rm intra}$ and ${\mathcal {I}}_{\rm inter}$ are the aggregate intra-cell and inter-cell interference, which will be discussed in
detail in Section III,
and $g_m$ denotes the power level unit in the $m$th retransmission  by adjusting the target received power at the BS equal to $g_m\rho$\cite{nan2018collision}\cite{1545847} (i.e., $g_1<g_2<...<g_m<...<g_J$). Note that $g_J$ is the maximum allowable power level unit.
% In this way, the access success probability in future transmissions can be improved over time.

\subsection{Problem Formulation and Objectives}
The URLLC requirement of the UL GF transmission is that the UEs can successfully complete their payload delivery within a limited time, i.e., $T_{\rm latency}\leqslant \cal T$, with a failure probability lower than a certain target, i.e., ${\cal P}_{F}\leqslant \varepsilon $. 
% Note that the GF access is failure, i.e., if 1) the received SINR at the BS is greater than the SINR threshold $\gamma_{th}$; and 2) no collision occurs (i.e., no other UEs successfully transmit to the typical BS simultaneously).
% Consequently, the outage probability is presented as
% \begin{align}\label{mth_data_transmission}
% &\mathcal{P}_{\rm out}
% %\\\nonumber
% = \sum\limits_{{n} = 0}^\infty  {\bigg\{{ \underbrace{{{\rm O}[{n},m]}}_{\rm I}
% \underbrace{\Theta^{\rm Krep} [{n},{K},m]}_{{\rm I}{\rm I}}
% \underbrace{{{\Big( {1 - \Theta^{\rm Krep} [{n},{K},m]} \Big)}^{{n}}}}_{{\rm I}{\rm I}{\rm I}}} \bigg\}},% \end{align}
% where collision events are detected by BS after it decodes the preambles in the step 1 of RA, and then no response will be fedback from the BS to the IoT devices, such that it can not proceed to the next step of RA [7,
% From the perspective of a UE, with a predefined latency constraint $\cal T$, we can define the reliability of a communication setup as the probability that the latency does not exceed this deadline ($T_{\rm latency}\leqslant \cal T$), and outage as the probability that it does ($P_{\rm out}\leqslant \varepsilon $).
For its performance characterization, we define the \textit{latent access failure probability}
as ${\cal P}_{F}(T_{\rm latency}\leqslant \cal T)$.
%From this perspective, it is meaningful to consider a probabilistic QoS in the following form
To address this inherent limitation of URLLC requirements, it is meaningful to consider a probabilistic Quality of Service (QoS) in the following form:
\begin{definition}
(URLLC QoS). We say the URLLC QoS of the
UE is satisfied in a given frame if:
%The QoS Level of the user is satisfied if 
%The latent access failure probability is defined as 
\begin{align}\label{definition}
{\cal P}_{F}(T_{\rm latency}\leqslant \cal T) \leqslant \varepsilon.
\end{align}
In 5G specification, $\varepsilon=10^{-5}$ and ${\cal T}=1$ ms for general URLLC requirement\cite{3gpp2018study}.
\end{definition}

% Note that the outage occurs if 1) the received SINR at the BS is lower than the SINR threshold $\gamma$, so the BS cannot decode the UE; 2)  the collision occurs (i.e., two or more UEs successfully transmit to the typical BS simultaneously).

% Here, we discuss the applicability of improved open loop power control for GF transmissions. 
% The considered power control is given by: P[dBm] = min $\{P_{max}, P_0 + 10 log10(M) + \alpha PL + g(k)\}$; where $P_{max}$ is the maximum
% transmit power, $P_0$ is the target receive power per resource block, M is the number of resource blocks, $\alpha$ is
% the fractional path loss compensation factor, $PL$ is the path loss and the function $g(k)$ gives a power boost
% step for the $k$th transmission.

% \subsection{Problem Formulation and Objectives}

% For comparison, we calculate the E2E delays to
% obtain quantitative results
%In order to calculate the HARQ latency time $T_{HA}$ and the outage probability $P_{\rm out}$, let us define the following variables:$K$: transmission times$T$: time duration of one TTI

\section{Performance Analysis and Evaluation}

This section presents a general analytical model for three different GF HARQ schemes.
% According to (\ref{definition}), the latent access failure probability depends on the latency constraint $\cal T$. 
We perform the analysis on a randomly chosen active UE in terms of the latent access failure probability under different latency constraints for three different GF schemes, respectively, in the following.

\subsection{Reactive scheme}
In the Reactive scheme as illustrated in Fig. 2, the latent access failure probability remains unchanged at the beginning of each HARQ round trip and only changes at the end of each HARQ round trip (i.e., after processing the feedback at the UE in the 4th TTI of this round trip), as the UE needs time to transmit packet and receive feedback.
For example, in one HARQ round trip (e.g., the $m$th round trip), for ${\cal T}=(m-1)T^{\rm RTT}_{\rm Reac}+2,(m-1)T^{\rm RTT}_{\rm Reac}+3,(m-1)T^{\rm RTT}_{\rm Reac}+4$ TTIs, the latent access failure probabilities are the same as ${\cal P}_{F}^{}[T_{\rm latency}\leqslant{\cal T}-1]$, since the UE can not receive feedback on time; for ${\cal T}=(m-1)T^{\rm RTT}_{\rm Reac}+5=mT^{\rm RTT}_{\rm Reac}+1$ TTIs, the latent access failure probability ${\cal P}_{F}^{}[T_{\rm latency}\leqslant{\cal T}]$ changes determined by the UE's retransmission or not after receiving NACK or ACK, respectively.

In order to calculate the latent access failure probabilities under various latency constraints, we need to know the maximum number $M$ of HARQ round trips allowed under the latency constraint $\cal T$ TTIs.
%\textcolor{red}{\footnote{Here, $\cal T$ means $\cal T$ (TTIs), if  }}.
For ease of presentation, we define
\begin{align}\label{M_reac}
M=\lfloor {({\cal T}-1)}/T^{\rm RTT}_{\rm Reac} \rfloor,   
\end{align}
%to imply the maximum number of HARQ round trips allowed under the latency constraint $\cal T$ 
with $T^{\rm RTT}_{\rm Reac}=4$ TTIs given in (\ref{RTT_reactive}).

% when $\cal T=$ 6, 7, and 8 TTIs, $M=1$ and the latent access failure probabilities ${\cal P}_{\rm out}^{}[T_{\rm latency}\leqslant{\cal T}]$ remain unchanged.

Note that inactive UEs (with empty data buffer) do not transmit, such that  they  do  not generate interference.
A UE is still active in the $m$th ($1\leqslant m\leqslant M $) round trip if none of its GF access in the last ($m-1$) round trips are successful.
%As mentioned before, whether  a UE  is  active  or  not in  the $m$th RTT depends on that whether its data transmission is successful in the $m-1$th RTT.
Mathematically, the active probability $\mathcal{A}_m^{}$ of the UE in the $m$th round trip, 
% which are decided by ${\cal P}_{\rm out}^{}[T_{\rm latency}\leqslant (m-1)T^{\rm RTT}_{\rm Reac}+1]$.  
is obtained as
% \begin{align}\label{active probability}
% {\mathcal{A}_m^{}} = \bigg\{ \begin{gathered}
%   {\cal P}_{\rm out}^{m}[T_{\rm latency}\leqslant{\cal T}_{}], m>1, \hfill \\
%   1, m = 1. \hfill \\ 
%   0, m = 0. \hfill \\
% \end{gathered}
% \end{align}
% \begin{align}\label{active_probability}
%  {\mathcal{A}_m^{}}=
%   \begin{cases}
%   {\cal P}_{\rm out}^{}[T_{\rm latency}\leqslant(m-1)T^{\rm RTT}_{\rm Reac}+1], 
%   &\mbox{if  $m \geqslant 1$},\\
% %   1, 
% %   &\mbox{if  $ m=1$}\\
%   1, 
%   &\mbox{if  $ m=0$ }.
%   \end{cases}
% \end{align}
% \begin{align}\label{active_probability}
%  {\mathcal{A}_m^{}}=
%   1-{\cal P}_{F}^{}[T_{\rm latency}\leqslant(m-1)T^{\rm RTT}_{\rm Reac}+1],
% \end{align}
\begin{align}\label{active_probability}
 {\mathcal{A}_m^{}}=
  1-{\cal P}_{F}^{}[T_{\rm latency}\leqslant T_{\rm{Reac}}[m-1]],
\end{align}
with $T_{\rm{Reac}}[m-1] $ obtained from \eqref{latency_rea1}.
% which means the probability that the UE fails to access in the $(m-1)$ round trips and needs to perform retransmission in the $m$th round trip.
% to derive the latent access failure probability  ${\cal P}_{\rm out}^{}[T_{\rm latency}\leqslant{\cal T}]$ with the maximum $M$ number of HARQ round trips, we need to  derive  

Based on (\ref{active_probability}), the latent access failure probability of a randomly chosen UE with the Reactive scheme is derived in the following \textbf{Theorem 1}.
\begin{theorem}
	The latent access failure probability of a randomly chosen UE with the Reactive HARQ scheme under the latency constraint ${\cal T}$ ${\rm TTIs}$ is derived as
	\begin{align}\label{reac_out}
	{\cal P}_{F}^{\rm Reac}[T_{\rm latency}\leqslant{\cal T}]=
	\begin{cases}
  1, 
  &\mbox{$M=0$},\\
%   1, 
%   &\mbox{if  $ m=1$}\\
  1-\sum\limits_{m = 1}^M{\mathcal{A}_m^{\rm Reac}}{{{\cal P}_m^{\rm Reac}}}, 
  &\mbox{$M\geqslant 1$},
  \end{cases}
	\end{align}
where $M$ is given in (\ref{M_reac}),
$\mathcal{A}_m^{\rm Reac}$ is given according to (\ref{active_probability}) as
\begin{align}\label{active_reac}
	\mathcal{A}_m^{\rm Reac}=
	\begin{cases}
  1, 
  &\mbox{$m=1$},\\
%   1, 
%   &\mbox{if  $ m=1$}\\
  1-\sum\limits_{i = 1}^{m-1}{\mathcal{A}_i^{\rm Reac}}{{{\cal P}_i^{\rm Reac}}}, 
  &\mbox{$ m\geqslant 2$},
  \end{cases}
	\end{align}
% \begin{align}
% \mathcal{A}_m^{\rm Reac}=1-\sum\limits_{i = 1}^{m-1}{\mathcal{A}_i^{\rm Reac}}{{{\cal P}_i^{\rm Reac}}}    
% \end{align},
and
${\cal P}_m^{\rm Reac}$ is the GF access success probability of the typical UE in the $m$th round trip with the Reactive scheme that derived in (\ref{mth_data_REAC}) of the following \textbf{Lemma 1}.
\end{theorem}
\begin{proof}
	See Appendix A 
\end{proof}

\begin{lemma}
The GF access success probability of the typical UE in the $m$th round trip with the Reactive scheme is given by
\begin{align}\label{mth_data_REAC}
&\mathcal{P}_m^{\rm Reac} 
%\\\nonumber
= \sum\limits_{{n} = 0}^\infty  {\bigg\{{ \underbrace{{{\rm O}[{n},m]}}_{\rm I}
\underbrace{\Theta^{\rm Reac} [{n},{}m]}_{{\rm I}{\rm I}}
\underbrace{{{\Big( {1 - \Theta^{\rm Reac} [{n},{}m]} \Big)}^{{n}}}}_{{\rm I}{\rm I}{\rm I}}} \bigg\}},
\end{align}
where
\begin{align}\label{mth_condition_n_REAC}
{\rm O}[{n},m]=\frac
{{{c^{(c + 1)}}\Gamma (n + c + 1){{\big( {\displaystyle%\frac
{{{\mathcal{A}_m^{\rm Reac}}{\lambda _a}}}/{{{\lambda _B}}}} \big)}^{n}}}}{{\Gamma (c + 1)\Gamma (n + 1){{\big( {\displaystyle{{{\mathcal{A}_m^{\rm Reac}}{\lambda _a}}}/{{{\lambda _B}}} + c} \big)}^{n + c + 1}}}},
\end{align}
and
\begin{align}\label{trans_reac}
&\Theta^{\rm Reac} [{n},{}m]=
 \exp \big( {- \displaystyle\frac{{\gamma_{\rm th}}\sigma ^2}{g_m{\rho }}\big)}{{{(1 + {{\gamma_{\rm th}}})}^{-n}}}\nonumber\\
&\times{{\exp \Big(
%\displaystyle\frac{{ - k{\gamma _{th}}{\sigma ^2}}}{\rho } 
{ - {{( {{{\gamma_{\rm th}}}} )}^{\frac{1}{2 }}}\displaystyle{{{}{\mathcal{A}_m^{\rm Reac}{\lambda _a}/{\lambda _B}}}}{{}}\arctan({{({}{{\gamma_{\rm th}}})}^{\frac{1}{2 }}}) }\Big)}}, (g_m\leqslant g_J).
\end{align}
Part I is the probability of the number of intra-cell interfering UEs for a typical BS $N=n${\footnote{{Note that $N=n$ means there are $n$ number of intra-cell interfering UEs (without the typicals UE),  i.e., $n+1$ number of active UEs in one cell.}}} derived following \cite[Eq.(3)]{yu2013downlink},
where c = 3.575 is a constant related to the approximate PMF of the PPP Voronoi cell and $\Gamma (\cdot)$ is the gamma function.
Part II is the transmission success probability of the UE conditioning on $N = n$.
Part III is the transmission failure probability that the transmissions  from other $n$ intra-cell interfering UEs are not successfully received by the BS, i.e., the non-collision probability of the UE.

% \begin{align}\label{mth_transmission}
% &\Theta^{\rm Reac} [{n},{},m]=
% \exp \big( {\displaystyle{{ - {\gamma _{th}}{\sigma ^2}}}{\rho }\big)}
% \\\nonumber
% &\displaystyle\frac{{\exp \Big(
% %\displaystyle\frac{{ - k{\gamma _{th}}{\sigma ^2}}}{\rho } 
% { - 2{{( {{\gamma _{th}}} )}^{\frac{2}{\alpha }}}\displaystyle\frac{{{}{\mathcal{A}_m^{\rm Krep}\lambda _a}}}{{{\lambda _B}}}\int_{{{( {{\gamma _{th}}} )}^{ - \frac{1}{\alpha }}}}^\infty  {\Big[ {1 - {{\Big( {\displaystyle\frac{1}{{1 + {x^{ - \alpha }}}}} \Big)}^k}} \Big]} xdx} \Big)}}{{{(1 + {\gamma _{th}})}^{kn}}}.
% % \\\nonumber
% % &\displaystyle\frac{{\exp \bigg( {\displaystyle\frac{{ - r{\gamma _{th}}{\sigma ^2}}}{\rho } - 2{{( {{\gamma _{th}}} )}^{\frac{2}{\alpha }}}\displaystyle\frac{{{}{\mathcal{A}_m^{k_{rep}}\lambda _a}}}{{{\lambda _B}}}\int_{{{( {{\gamma _{th}}} )}^{ - \frac{1}{\alpha }}}}^\infty  {\Big[ {1 - {{\Big( {\displaystyle\frac{1}{{1 + {x^{ - \alpha }}}}} \Big)}^r}} \Big]} xdx} \bigg)}}{{{{(1 + {\gamma _{th}})}^{rn}}}}.
% \end{align}
\end{lemma}
\begin{proof}
See Appendix B.
\end{proof}

\begin{remark}
In (\ref{trans_reac}), it can be shown that the transmission success probability (II in  (\ref{mth_data_REAC})) of the typical UE is inversely proportional to the received {\rm SINR} threshold ${\gamma_{\rm th}}$ and the density ratio ${\lambda_a}/{\lambda_B}$. 
The  transmission failure probabilities of other interfering UEs (III in (\ref{mth_data_REAC})) (i.e, the non-collision probability of the typical UE) are directly proportional to the received SINR threshold ${\gamma_{\rm th}}$ and the density ratio ${\lambda_a}/{\lambda_B}$. 
Therefore, a tradeoff between transmission success probability and non-collision probability is observed.
\end{remark}

\subsection{K-repetition scheme}
In the K-repetition scheme as illustrated in Fig. 3, the latent access failure probability also changes  at  the  end  of  each  HARQ round trip similar to the Reactive scheme, but with longer round trip time $T^{\rm RTT}_{\rm Krep}$ TTIs given in (\ref{RTT_krep}).
More specifically, in one HARQ round trip (e.g., the $m$th round trip) of the K-repetition scheme, for ${\cal T}= (m-1)T^{\rm RTT}_{\rm Krep}+2, (m-1)T^{\rm RTT}_{\rm Krep}+3, ... , (m-1)T^{\rm RTT}_{\rm Krep}+K_{\rm Krep}+3$ TTIs, the latent access failure probabilities are the same as ${\cal P}_{F}^{}[T_{\rm latency}\leqslant{\cal T}-1]$; for ${\cal T}=(m-1)T^{\rm RTT}_{\rm Krep}+K_{\rm Krep}+4=mT^{\rm RTT}_{\rm Krep}+1$ TTIs, the latent access failure probabilities ${\cal P}_{F}^{}[T_{\rm latency}\leqslant{\cal T}]$ changes determined by the UE’s retransmission or not after receiving NACK or ACK, respectively.
% when ${\cal T}=K+5,..., K+8$ TTIs, the latent access failure probabilities ${\cal P}_{\rm out}^{}[T_{\rm latency}\leqslant{\cal T}]$ remain unchanged.
Let us define
\begin{align}\label{M_krep}
M=\lfloor {({\cal T}-1)}/T^{\rm RTT}_{\rm Krep} \rfloor,   
\end{align}
to imply the maximum number of HARQ round trips allowed under the latency constraint $\cal T$ {TTIs} with $T^{\rm RTT}_{\rm Krep}=K_{\rm Krep}+3$ TTIs, we can derive the latent access failure probability of a randomly chosen UE with the K-repetition scheme in the following \textbf{Theorem 2}.
\begin{theorem}
	The latent access failure probability of a randomly chosen UE with the K-repetition scheme under latency constraint ${\cal T}$ {\rm TTIs} is derived as
	\begin{align}\label{krep_out}
	{\cal P}_{F}^{\rm Krep}[T_{\rm latency}\leqslant{\cal T}]=
	\begin{cases}
  1, 
  &\mbox{  $M=0$}\\
%   1, 
%   &\mbox{if  $ m=1$}\\
  1-\sum\limits_{m = 1}^M{\mathcal{A}_m^{\rm Krep}}{{{\cal P}_m^{\rm Krep}}}, 
  &\mbox{  $ M\geqslant 1$ },
  \end{cases}
	\end{align}
% 	\begin{align}\label{krep_out}
% 	{\cal P}_{\rm out}^{\rm Krep}[T_{\rm latency}\leqslant{\cal T}] =1-\sum\limits_{m = 1}^M{\mathcal{A}_m^{\rm Krep}}{{{\cal P}_m^{\rm Krep}}}, (M=\lfloor {({\cal T}-1)}/T^{\rm RTT}_{\rm Krep} \rfloor),
% 	\end{align}
where $M$ is given in  (\ref{M_krep}), $\mathcal{A}_m^{\rm Krep}$ is obtained according to (\ref{active_probability}) as
\begin{align}\label{active_krep}
	\mathcal{A}_m^{\rm Krep}=
	\begin{cases}
  1, 
  &\mbox{$m=1$},\\
%   1, 
%   &\mbox{if  $ m=1$}\\
  1-\sum\limits_{i = 1}^{m-1}{\mathcal{A}_i^{\rm Krep}}{{{\cal P}_i^{\rm Krep}}}, 
  &\mbox{$ m\geqslant 2$},
  \end{cases}
	\end{align}
% 	\begin{align}\label{active_probability_krep}
%   {\mathcal{A}_m^{}}=
%   {\cal P}_{\rm out}^{}[T_{\rm latency}\leqslant(m-1)T^{\rm RTT}_{\rm Krep}+K+1], 
%   \end{align}
and ${\cal P}_m^{\rm Krep}$ is the GF access success probability of the typical UE in the $m$th round trip with the K-repetition scheme that derived in (\ref{access_krep}) of the following \textbf{Lemma 2}.
\end{theorem}

\begin{lemma}
The GF access success probability of the typical UE in the $m$th HARQ round trip with the K-repetition scheme is derived as
\begin{align}\label{access_krep}
&\mathcal{P}_m^{\rm Krep}= 
\nonumber \\
& \sum\limits_{{n} = 0}^\infty  {\bigg\{{ \underbrace{{{\rm O}[{n},m]}}_{\rm I}
\underbrace{\Theta^{\rm Krep} [{n},{m},K_{\rm Krep}]}_{{\rm I}{\rm I}}
\underbrace{{{\Big( {1 - \Theta^{\rm Krep} [{n},{m},K_{\rm Krep}]} \Big)}^{{n}}}}_{{\rm I}{\rm I}{\rm I}}} \bigg\}},
\end{align}
where
\begin{align}\label{mth_condition_n}
{\rm O}[{n},m]=\frac
{{{c^{(c + 1)}}\Gamma (n + c + 1){{\big( {\displaystyle%\frac
{{{\mathcal{A}_m^{\rm Krep}}{\lambda _a}}}/{{{\lambda _B}}}} \big)}^{n}}}}{{\Gamma (c + 1)\Gamma (n + 1){{\big( {\displaystyle{{{\mathcal{A}_m^{\rm Krep}}{\lambda _a}}}/{{{\lambda _B}}} + c} \big)}^{n + c + 1}}}},
\end{align}
and
\begin{align}\label{trans_krep1}
&\Theta^{\rm Krep} [{n},{m},K_{\rm Krep}]\nonumber \\
& =\sum\limits_{{k} = 1}^{{K_{\rm Krep}}} {{( - 1)}^{{k} + 1}}\Big( \begin{array}{l}
{K_{\rm Krep}}\\
{k}
\end{array} \Big)
\exp \big( {- \displaystyle\frac{{k\gamma_{\rm th}}\sigma ^2}{g_m{\rho }}\big)}{{{(1 + {{\gamma_{\rm th}}})}^{-kn}}}
 \nonumber \\
% \exp \big( {\displaystyle{{ - k{\gamma _{th}}{\sigma ^2}}}/{\rho }\big)}
&\times\displaystyle{{\exp \Big(
		%\displaystyle\frac{{ - k{\gamma _{th}}{\sigma ^2}}}{\rho } 
		{ - \displaystyle{{{}{\mathcal{A}_m^{\rm Krep}\lambda _a}}}/{{{\lambda _B}}}\Big( {{}_2{F_1}\Big( { - \frac{2}{\alpha },k;\frac{{\alpha  - 2}}{\alpha }; - {{\gamma_{\rm th}}}} \Big) - 1} \Big)} \Big)}}.
\end{align}
{Similar to \rm{\textbf{Lemma 1.}},} Part I is the probability of the number of intra-cell interfering UEs $N=n$.
Part II is the transmission success probability of the UE conditioning on $N = n$.
Part III is the {non-collision probability of the UE.}
\end{lemma}
\begin{proof}
See Appendix C.
\end{proof}

%As the repetition value $K$, the path loss exponent $\alpha$, and the transmit distance $x$ are positive, 
\begin{remark}
It is evident from (\ref{trans_krep1}) that the transmission success probability (II in \eqref{access_krep}) of the typical UE increases, whereas the non-collision probability (III in \eqref{access_krep}) decreases with increasing the repetition value $K_{\rm Krep}$.
Therefore, there exists a tradeoff between transmission success probability and non-collision probability.
For illustration, the relationship among GF access success probability {$(\mathcal{P}_1^{\rm Krep})$}, the transmission success probability {$(\mathcal{P}_1^{\rm Krep}$ with III=1)}, and the non-collision probability {$(\mathcal{P}_1^{\rm Krep}$ with II=1)} {versus repetition values} are shown in Fig. 5.
We can see that in certain scenario in (b) (i.e., ${\gamma_{\rm th}}=-10$ dB and ${\lambda_D}/{\lambda_B}>4\times10^4$), the increase of repetition value $K_{\rm Krep}$ could not further improve, and even degrades the GF access success probability.
This is due to the fact that increasing the repetition increases the collisions in overloaded traffic scenario, and wastes extra time and frequency resource. Further details will be described later in Section V.
\end{remark}
\begin{figure}
\centering
\subfigure[]
{\includegraphics[width=3.2in,height=2.7in]{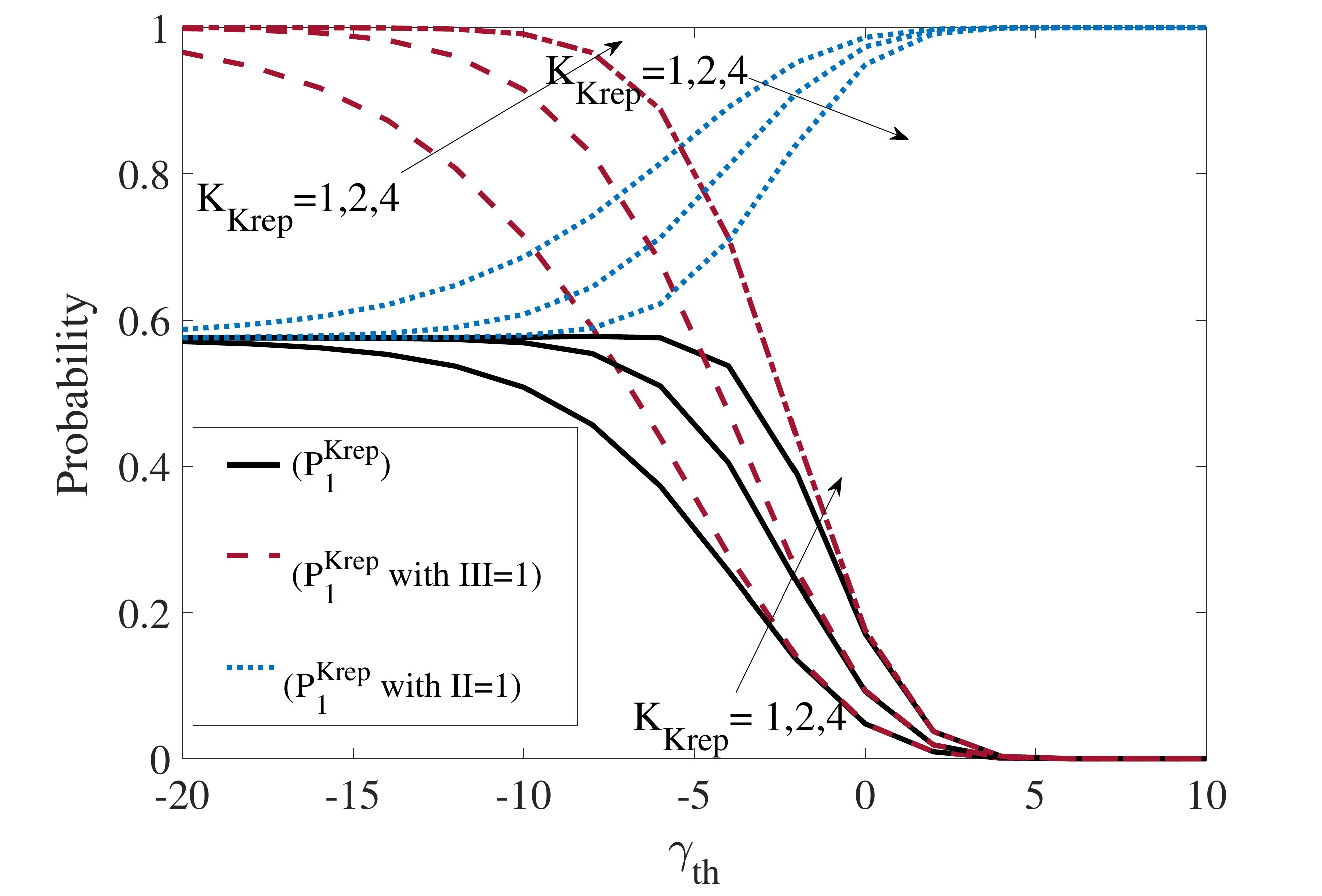}}
\subfigure[]
{\includegraphics[width=3.2in,height=2.7in]{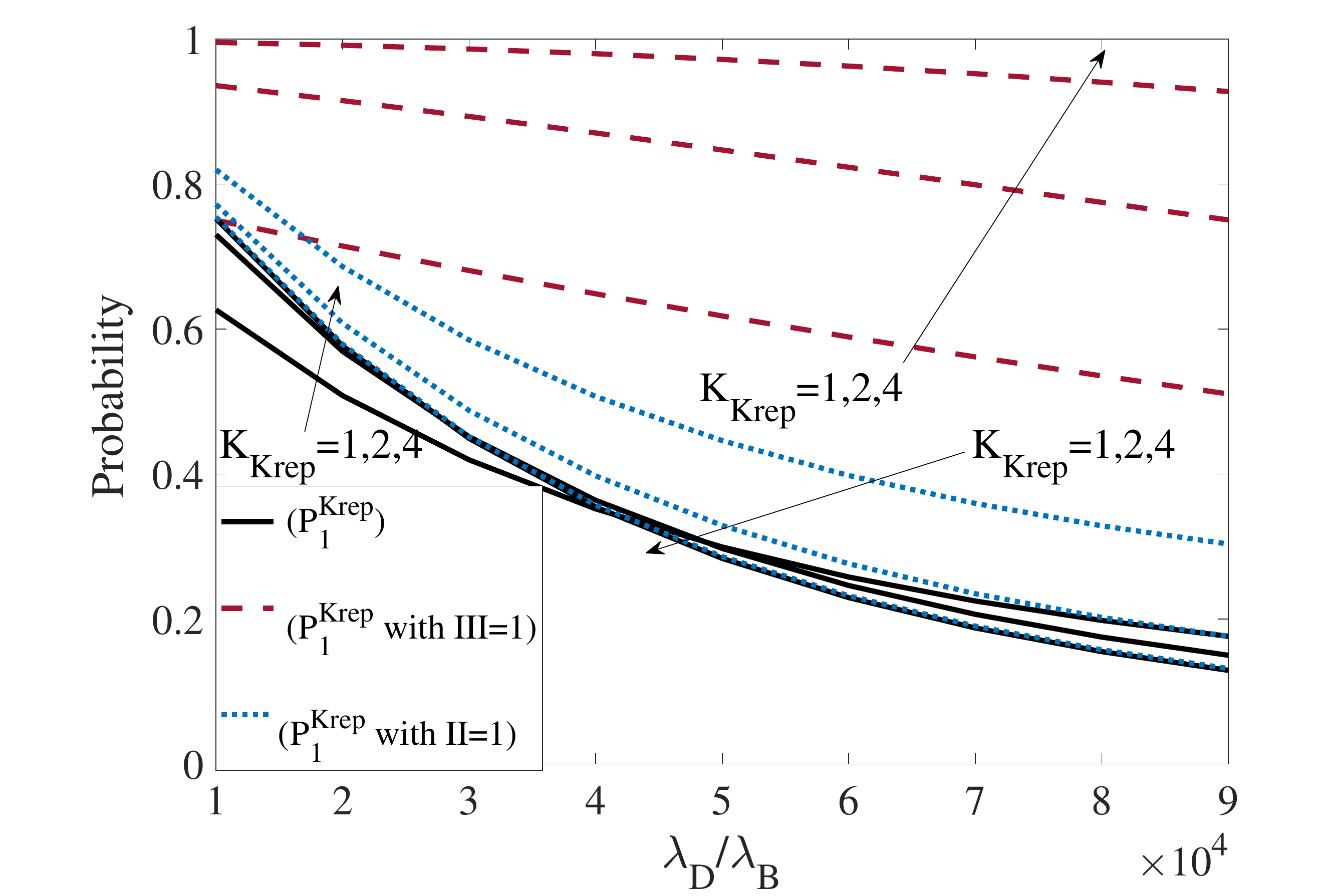}}
\caption{Comparing GF access success probability (${\cal P}_1^{\rm Krep}$), transmission success probability (${\cal P}_1^{\rm Krep}$ with III = 1), and non-collision probability (${\cal P}_1^{\rm Krep}$ with II = 1).
The parameters are $\lambda_B=1$ BS/km$^2$, $\lambda_D=20000$ UEs/km$^2$, $p_a=0.0011$, $\rho= −130$ dBm, $\gamma_{\rm th}=-10$ dB and $\sigma^2=−126.2$ dBm.}
\label{fig5}
\end{figure}
% The latent access failure probabilities under various latency constraints of the K-repetition scheme can be derived  following  like the iteration process shown in the flowchart in Fig. 5.

Finally, the latent access failure probabilities under arbitrary latency constraints of a randomly chosen UE with the K-repetition and Reactive schemes can be derived based on the iteration process. Note that the Reactive scheme is a special case of
K-repetition scheme when the repetition value $K_{\rm Krep}=1$.
\begin{figure}
	\centering
	\includegraphics[width=3.6in,height=3.3in]{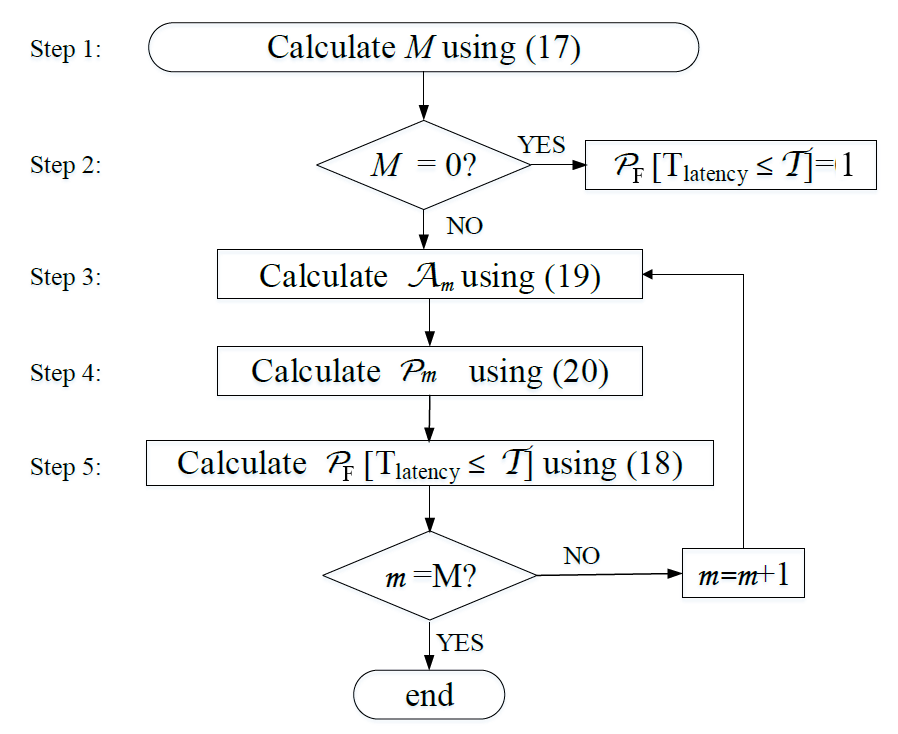}
	\caption{Flowchart for deriving the latent access failure probability of the K-repetition scheme and the Reactive scheme.}
	\label{fig:my_label}
\end{figure}
We assume $m$ is a variable that denotes the HARQ round trip from 1 to $M$. 
The iteration process for calculating the latent access failure probability is
shown in Fig. 6. Details of this process are described by the following:
\begin{itemize}
\item Step 1: Calculate the maximum number of HARQ round trips $M$ under the given latency constraint ${\cal T}_{}$ {TTIs} using (\ref{M_krep}). 

\item Step 2: If $M=0$, $
{\cal P}_{F}^{\rm Reac}[T_{\rm latency}\leqslant{\cal T}]=1$, otherwise  go to Step 3;
% Initialize the latency constraint $\cal T$=1 TTI, and $
% {\cal P}_{F}^{\rm Reac}[T_{\rm latency}\leqslant{1}]=0$;
\item Step 3: Calculate the active probability ${\mathcal{A}_m^{}}$ in the $m$th HARQ round trip using (\ref{active_krep});
\item Step 4: Calculate the GF access success probability in the $m$th round trip using (\ref{access_krep});
\item Step 5: Calculate the latent access failure probability in the $m$th round trip using (\ref{krep_out}).
% \item Step 5: Repeat the step 3 to 4 until $m=M$ and calculate the latent access failure probability $
% {\cal P}_{\rm out}^{\rm Reac}[T_{\rm latency}\leqslant{\cal T}]$ using (\ref{krep_out}).
\end{itemize}
Repeating Step 3 to 5 until $m=M$, the latent access failure probability under latency constraint ${\cal T}_{}$ can be obtained.

\subsection{Proactive scheme}
The analytical model for the Proactive scheme is more complicated compared with the Reactive and K-repetition schemes.
In the former two schemes, the latent access failure probabilities only change at the end of each HARQ round trip, as the BS processes the received signal and sends the feedback to the UE once in each round trip.
However, in the Proactive scheme, the latent access failure probabilities change at several TTIs in one round trip, as the BS processes each repetition and sends the feedback to the UE at several TTIs.
% Take one example, as the Proactive scheme with the maximum $K$ repetitions shown in Fig. 4, in one HARQ round trip (e.g., the $m$th round trip), for ${\cal T}= 2, 3$ and $4$ TTIs, the latent access failure probabilities do not update with ${\cal P}_{\rm out}^{}[T_{\rm latency}\leqslant{\cal T}]=1$; when ${\cal T}=5,...,K+4$ TTIs, the latent access failure probabilities ${\cal P}_{\rm out}^{}[T_{\rm latency}\leqslant{\cal T}]$ update. 
% when ${\cal T}=K+5,..., K+7$ TTIs, the latent access failure probabilities ${\cal P}_{\rm out}^{}[T_{\rm latency}\leqslant{\cal T}]$ remain unchanged
Due to the complexity of the Proactive scheme, we first analyze the latent access failure probability of a randomly chosen UE with the latency constraint ${\cal T}\leqslant K_{\rm Proa}+4$ TTIs without HARQ retransmissions. 
\subsubsection{Proactive scheme without HARQ retransmissions, ${\cal T}\leqslant K_{\rm Proa}+4$ {\rm TTIs}}

Compared with the K-repetition scheme, in which the UE is enforced to perform $K_{\rm Krep}$ repetitions no matter if its transmission is successful or not within $K_{\rm Krep}$ times, the UE in the Proactive scheme is allowed to terminate the repetition once the UE receives ACK.
Take one example, as shown in Fig. 7, the UE-1  successfully transmits the packet in the 1st repetition, the UE-1 knows the success of its 1st repetition in the 4th repetition, and the UE-1 terminates its 5th repetition. 
That is to say, if a UE does not have a second packet to be transmitted, the Proactive scheme could help to reduce its interference to other UE(s) that share the same resource and happen to be active at the same time.

\begin{figure}
	\centering
	\includegraphics[width=3.3in,height=1.4in]{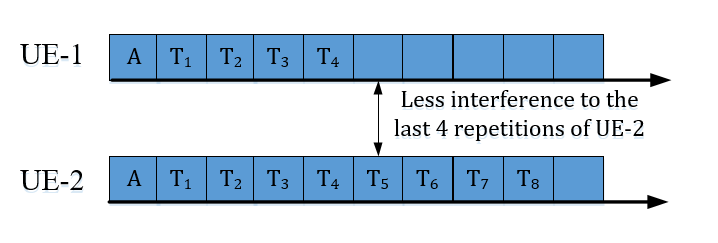}
	\caption{Early termination reduces UE interference }
	\label{fig:my_label}
\end{figure}

Due to the fact that the ACK/NACK feedback can only be received after 3TTIs, for the maximum repetition value $K_{\rm Proa}\leqslant4$, the UE can not receive feedback before completing $K_{\rm Proa}$ repetitions, thus the UE needs to  complete all the $K_{\rm Proa}$ repetitions without terminating earlier and the
number of interfering users will not change in each repetition of one round trip.
% For $K_{\rm Proa}\leqslant4$, the latent access failure probabilities for the Proactive scheme is the same as that of  the K-repetition scheme in (\ref{krep_out}) of Section III-B.

For $K_{\rm Proa}\geqslant5$, the UE can receive feedback from the BS to determine retransmission or not.
For instance, the ACK feedback decreases the number of interfering users in the later repetitions as shown in Fig. 7.
Let us denote that the $1$st successful transmission occurs in the $l$th repetition, thus the feedback of this repetition will be received in the ($l+3$)th repetition, which affects the latent access failure probability of the UE in the ($l+3$)th repetition, and from the ($l+4$)th repetition, these successful UEs will not repeat the rest ($K_{\rm Proa}-l-3$) repetitions for this packet any more.
% Thus, the UE will update the latency outage probability in the ($l+3$)th repetition.
% These successfully transmitted UEs in the $l$th repetition will not transmit from ($l+4$)th repetition.

We define the \textit{feedback factor} for the $l$th ($1\le l \le K_{\rm Proa}$) repetition as $\eta_{1,l}$, which means the GF access failure probability in the former ($l-4$) repetitions\footnote{Note that the ACK/NACK feedback can only be received  after 3TTIs, thus the feedback from the former ($l-4$) repetitions will affect the $l$th repetition. Only the failure UEs in the former ($l-4$) repetitions will transmit in the $l$th repetition.}.
It is obvious that $\eta_{1,l}=1$ when $1\le l \le 4$.
%As shown in Fig. 2, we note that the UE could receive the first feedback three TTIs after the  first transmission. 
%That is to say, when $K\leqslant4$, no feedback will be received to affect the number of interfering UEs in each TTI.
%Thus, when $K\leqslant4$, the latent access failure probability ${\cal P}^{\rm out}_{\rm Proa}[T_{\rm latency}>{\cal T}_{K}]$ with latency target ${\cal T}_{K}=4+K$ TTIs is defined as Theorem 1.
%A UE still needs to transmit in the $l$th repetition if all the $l-4$ ($l>4$) repetitions before are not successful.
Then we derive the feedback factor as
% \begin{align}\label{feedback1}
% % {\mathcal{A}_m^{}} = \bigg\{ \begin{gathered}
% %   {\cal P}_{\rm out}^{m}[T_{\rm latency}\leqslant{\cal T}_{}], m>1, \hfill \\
% %   1, m = 1. \hfill \\ 
% %   0, m = 0. \hfill \\
% % \end{gathered}
% % \end{align}
\begin{align}\label{feedback1}
 {\eta_{1,l}}=
  \begin{cases}
  1, 
  &\mbox{  $1 \leqslant l \leqslant 4$},\\
%   1, 
%   &\mbox{if  $ m=1$}\\
  1-{\cal P}_{1, l-4}^{\rm Proa}, 
  &\mbox{  $ l\geqslant 5$ },
  \end{cases}
\end{align}
where ${\cal P}_{1,l}^{\rm Proa}$ is derived in the following \textbf{Lemma 3}.

\begin{lemma}
We define the transmission success probability in the $l$th repetition as ${\mathbb P}_{1,l}$, the transmission success probability in all $l$ repetitions as ${\Theta^{\rm Proa} [{n},{1},l]}$ (i.e., any one of the $l$ repetitions succeeds (${\mathbb P}_{1,l}$)), and the access success probability in $l$ repetitions as ${\cal P}_{1,l}^{\rm Proa}$ (considering collision).
Then, the GF access success probability of a randomly chosen UE with the Proactive scheme under the latency constraint ${\cal T}\leqslant K_{\rm Proa}+4$ {\rm TTIs} is driven as
\begin{align}\label{data_proa}
&\mathcal{P}_{1,l}^{\rm Proa} 
%\\\nonumber
= \sum\limits_{{n} = 0}^\infty  {\bigg\{{ \underbrace{{{\rm O}[{n},1,l]}}_{\rm I}
\underbrace{\Theta^{\rm Proa} [{n},{1},l]}_{{\rm I}{\rm I}}
\underbrace{{{\Big( {1 - \Theta^{\rm Proa} [{n},{1},l]} \Big)}^{{n}}}}_{{\rm I}{\rm I}{\rm I}}} \bigg\}},
\end{align}
where
\begin{align}\label{N_condition1mu_proa}
{{\rm O}[{n},1,l]} = \displaystyle\frac{{{c^{(c + 1)}}\Gamma (n + c + 1){{\big( {{{{\eta_{1,l}\lambda _{a}}}}/{{{\lambda _B}}}} \big)}^{n}}}}{{\Gamma (c + 1)\Gamma (n+ 1){{\big( {{{{\eta_{1,l}\lambda _{a}}}}/{{{\lambda _B}}} + c} \big)}^{n + c + 1}}}}, 
\end{align}
and for $l\le 4$,
\begin{align}\label{trans_l4}
    &{\Theta^{\rm Proa} [{n},{1},l]}=1-\prod\limits_{r = 1}^{l}(1-{\mathbb P}_{1,r})\\\nonumber
    &=\sum\limits_{{r} = 1}^{{l}} {{( - 1)}^{{r} + 1}}\Big( \begin{array}{l}
{l}\\
{r}
\end{array} \Big)
% \exp \big( {\displaystyle{{ - k{\gamma _{th}}{\sigma ^2}}}/{\rho }\big)}
\exp \big( {- \displaystyle\frac{{r\gamma_{\rm th}}\sigma ^2}{g_m{\rho }}\big)}{{{(1 + {{\gamma_{\rm th}}})}^{-kn}}}
\nonumber \\
&\times\displaystyle{{\exp \Big(
%\displaystyle\frac{{ - k{\gamma _{th}}{\sigma ^2}}}{\rho } 
{ - \displaystyle{{{}{{\eta_{1,r}}\lambda _a}}}/{{{\lambda _B}}}\Big( {{}_2{F_1}\Big( { - \frac{2}{\alpha },k;\frac{{\alpha  - 2}}{\alpha }; - {{\gamma_{\rm th}}}} \Big) - 1} \Big)} \Big)}},
\end{align}
and for $l\geqslant5$, 
\begin{align}\label{trans_l5}
{\Theta^{\rm Proa} [{n},{1},l]} =1-(1-{{\Theta^{\rm Proa} [{n},{1},4]}})\prod\limits_{r = 5}^{l}\big({1-\mathbb P}_{1,r}\big),
\end{align}
with
\begin{align}\label{m=1trans}
        {\mathbb P}_{1,r}={\eta_{1,r}}
        %\sum\limits_{n = 0}^\infty 
        {{{{\rm O}[{n},1,r]}{\Theta^{\rm Proa} [{n},{1},1]}}},
    \end{align}
where ${{\Theta^{\rm Proa} [{n},{1},1]}}$ is obtained from (\ref{trans_l4}), and ${{\rm O}[{n},1,r]}$  is obtained from (\ref{N_condition1mu_proa}).
\end{lemma}
\begin{proof}
See Appendix D.
\end{proof}

%and the GF success probability in the $1st$ RT as ${\mathbb P}_{\rm Proc}^1(l)$. 
%We note that as ${\mathbb P}_1$=${\mathbb P}_2$=${\mathbb P}_3$=${\mathbb P}_4$, then ${\mathbb P}_{1+4}$=${\mathbb P}_{2+4}$=${\mathbb P}_{3+4}$=${\mathbb P}_{4+4}$.
%For ease of presentation, we assume $K=4l-3+m$ ($m=0,1,2,3$, $l=1,2,...$).

In order to calculate the latent access failure probabilities under arbitrary latency constraints ${\cal T}\leqslant K_{\rm Proa}+4$ {\rm TTIs}, we define two indexes for $\cal T$ as
\begin{align}\label{index}
 \begin{cases}
  \mu=\lfloor {({\cal T}-2)}/T^{\rm RTT}_{{\rm Proa},K,0} \rfloor,
  &\mbox{}\\
%   1, 
%   &\mbox{if  $ m=1$}\\
  \nu=\rm{mod}({\cal T}-2,T^{\rm RTT}_{{\rm Proa},K,0}), 
  &\mbox{}
  \end{cases}   
\end{align}
where $T^{\rm RTT}_{{\rm Proa},K,0}$ is given in (\ref{pro_rrt1})\footnote{In the Proactive scheme with $m$ HARQ round trips, a UE is still active in the $m$th ($1\leqslant m\leqslant M$) HARQ round trip if none of its GF access in the former ($m−1$) HARQ round trips is successful.
That is to say, all the maximum ${K_{\rm Proa}}$ repetitions in the Proactive scheme in the former ($m−1$) HARQ round trips are not successful, i.e., $l=0$.}, $\mu$ implies the maximum number of the HARQ round trips under the latency constraint (for the Proactive scheme under the latency constraint ${\cal T}\leqslant K_{\rm Proa}+4$, $\mu=0$), $\nu$ implies the updated TTI index for the latent access failure probability in each HARQ round trip.

Then, the latent access failure probability  of a randomly chosen UE with the Proactive scheme under the latency constraint ${\cal T}\leqslant K_{\rm Proa}+4$ is derived in  \textbf{Theorem 3}.
\begin{theorem}
The latent access failure probability  of a randomly chosen UE with the Proactive scheme under the latency constraint ${\cal T}\leqslant K_{\rm Proa}+4$ {\rm TTIs} is derived  as
% \begin{align}\label{pro_rrt}
% %\hspace{-0.5cm}
%  {\cal P}_{\rm out}^{\rm Proa}[T_{\rm latency}\leqslant{\cal T}_{}] =
% \begin{cases}
% 1-{\cal P}_K^{\rm Proa}, &\rm{if \ K \leqslant 4 \  with \ {\cal T}_{}=K+4},\\
% 1-{\cal P}_l^{\rm Proa},
% %-\sum\limits_{r = 5}^{l + 3} {\cal P}_{r}\eta_r, 
% &\rm{if \ 1\leqslant \textit{l} \leqslant K \ and \  K  \geqslant5 \ with \ {\cal T}_{1}=\textit{l}+4 },
% \end{cases}
% \end{align}
\begin{align}\label{pro_rrt}
 {\cal P}_{F}^{\rm Proa}[T_{\rm latency}\leqslant{\cal T}_{}] =
  \begin{cases}
  1, 
  &\mbox{  $ \nu \leqslant 2$, and $\mu=0$},\\
%   1, 
%   &\mbox{if  $ m=1$}\\
  1-{\cal P}_{1,\nu-2}^{\rm Proa}, 
  &\mbox{ $ \nu \geqslant 3$, and $\mu=0$ }.
  \end{cases}
\end{align}
% \begin{align}\label{pro_rrt}
%  {\cal P}_{\rm out}^{\rm Proa}[T_{\rm latency}\leqslant{\cal T}_{}] =
%   \begin{cases}
%   1-{\cal P}_K^{\rm Proa}, 
%   &\mbox{if  $ K\leqslant 4$ and ${\cal T}_{}=K+4$},\\
% %   1, 
% %   &\mbox{if  $ m=1$}\\
%   1-{\cal P}_l^{\rm Proa}, 
%   &\mbox{if $ 1\leqslant l \leqslant K $, $K \geqslant 5$  and $ {\cal T}_{}=l+4$ }.
%   \end{cases}
% \end{align}
where ${\cal P}_{1,\nu-2}^{\rm Proa}$ is obtained from (\ref{data_proa}) of \textbf{Lemma 3}.
% \begin{align}\label{proactive}
% &{\cal P}_{\rm out}^{Pro}[T_{\rm latency}>{\cal T}_K] =1-{\cal P}_{1,K}^{Pro},
% \end{align}
% where ${\cal P}_{1,K}^{Pro}$ is the GF access success probability of the UE in the $1$st RTT and is derived in the following Lemma.
\end{theorem}

Next, we extend the analysis of the latent access failure probabilities of the typical UE with the Proactive scheme to an arbitrary latency constraint $\cal  T$ allowing the maximum $M$ number of HARQ round trips.

\subsubsection{Proactive scheme with HARQ retransmissions}
% As mentioned before,  when ${\cal T}=K+3,..., K+5$ TTIs in the 1st RTT, the latent access failure probabilities ${\cal P}_{\rm out}^{}[T_{\rm latency}\leqslant{\cal T}]$ remain unchanged.
% Now we extend the analysis to an arbitrary latency constraint $\cal  T$ allowing the maximum $M$ number of HARQ round trips.
In the Proactive scheme with HARQ  retransmissions, a UE is still active in the $m$th ($1\leqslant m\leqslant M$) HARQ round trip if none of its GF access in the former ($m−1$) HARQ round trips are successful.
That is to say, all the maximum ${K_{\rm Proa}}$ repetitions in the Proactive scheme in the former ($m−1$) HARQ round trips are not successful.
Similar to the other two schemes, we give the active probability ${\mathcal{A}_m^{Proa}}$ in the $m$th HARQ round trip in \eqref{active_proa}.
% \begin{align}\label{non_empty_proac}
%  {\mathcal{A}_m^{Proa}}=
%   {\cal P}_{\rm out}^{}[T_{\rm latency}\leqslant (m-1)T^{\rm RTT}_{{\rm Proa},K}+1].     
% \end{align}
For an arbitrary latency constraint $\cal  T$ {\rm TTIs}, we first obtain the two indexes $\mu$ and $\nu$ using (\ref{index}),  i.e., the maximum number of the HARQ round trips under the latency constraint is $M=\mu$.
% \begin{align}\label{M_proa}
% M=\lfloor {({\cal T}-2)}/T^{\rm RTT}_{{\rm Proa},K} \rfloor, (M \ge 0)
% \end{align}
% where $T^{\rm RTT}_{\rm {\rm Proa},K}=K+3$.
% The non-empty probability ${\mathcal{A}}_m$ is obtained from (\ref{active_probability_krep}) in \textbf{Theorem 2}.
% For  $M=0$ (i.e., $T\le K_{\rm Proa}+4$ TTIs), the latency outage probabilities can be obtained using the flowchart in Fig. 8. 
% For $M\ge 1$, we need to derive the active probability in the $m$th HARQ round trip as
Then, the latent access failure probability can be obtained in the following \textbf{Theorem 4}.
\begin{theorem}
The latent access failure probability of a randomly chosen UE with the Proactive HARQ scheme under arbitrary latency constraint $\cal T$ {\rm TTIs} is derived as

\begin{align}\label{promulti_m}
&{\cal P}_{F}^{\rm Proa}[T_{\rm latency}\leqslant{\cal T}_{}]=\nonumber \\
&\begin{cases}
1, 
&\mbox{  $ \nu \leqslant 2$ $\&$ $\mu=0$},\\
%   1, 
%   &\mbox{if  $ m=1$}\\
1-{\cal P}_{1, \nu-2}^{\rm Proa}, 
&\mbox{ $ \nu \geqslant 3$ $\&$ $\mu=0$ },\\
1-\sum\limits_{m = 1}^{M}{\mathcal{A}_m^{\rm Proa}}{{{\cal P}_{m,K}^{\rm Proa}}}, 
&\mbox{ $ \nu \leqslant 2$ $\&$ $\mu\ge1$},\\
1-\sum\limits_{m = 1}^{M}{\mathcal{A}_m^{\rm Proa}}{{{\cal P}_{m,K}^{\rm Proa}}}+{\mathcal{A}_{M+1}^{\rm Proa}}{{{\cal P}_{M+1,\nu-2}^{\rm Proa}}} 
&\mbox{ $ \nu \geqslant 3$ $\&$ $\mu\ge1$},\\
\end{cases}
  \end{align}
where
$\mathcal{A}_m^{\rm Proa}$ is obtained according to (\ref{active_probability}) as
\begin{align}\label{active_proa}
	\mathcal{A}_m^{\rm Proa}=
	\begin{cases}
  1, 
  &\mbox{$m=1$},\\
%   1, 
%   &\mbox{if  $ m=1$}\\
  1-\sum\limits_{i = 1}^{m-1}{\mathcal{A}_i^{\rm Proa}}{{{\cal P}_i^{\rm Proa}}}, 
  &\mbox{$ m\geqslant 2$},
  \end{cases}
	\end{align}
and ${\cal P}_{m,l}^{\rm Proa}$ is the GF access probability of a typical UE in the $m$th HARQ round trip, given in the following \textbf{Lemma 4}.
% \begin{proof}
% % See Appendix E.
% \end{proof}
\end{theorem}

\begin{lemma}
The GF access success probability of a randomly chosen UE with the Proactive HARQ scheme in the $m$th HARQ round trip is driven as
\begin{align}\label{data_proa_m}
&\mathcal{P}_{m,l}^{\rm Proa} 
%\\\nonumber
= \sum\limits_{{n} = 0}^\infty  {\bigg\{{ \underbrace{{{\rm O}[{n},m,l]}}_{\rm I}
\underbrace{\Theta^{\rm Proa} [{n},{m},l]}_{{\rm I}{\rm I}}
\underbrace{{{\Big( {1 - \Theta^{\rm Proa} [{n},{m},l]} \Big)}^{{n}}}}_{{\rm I}{\rm I}{\rm I}}} \bigg\}},
\end{align}
where
\begin{align}\label{N_condition1mu_proa_m}
{{\rm O}[{n},m,l]} = \displaystyle\frac{{{c^{(c + 1)}}\Gamma (n + c + 1){{\big( {{{{\eta_{m,l}{\mathcal{A}_m^{Proa}}\lambda _{a}}}}/{{{\lambda _B}}}} \big)}^{n}}}}{{\Gamma (c + 1)\Gamma (n+ 1){{\big( {{{{\eta_{m,l}{\mathcal{A}_m^{Proa}}\lambda _{a}}}}/{{{\lambda _B}}} + c} \big)}^{n + c + 1}}}},  
\end{align}
with
\begin{align}\label{feedbackm_proac}
 {\eta_{m,l}}=
  \begin{cases}
  1, 
  &\mbox{if  $1 \leqslant l \leqslant 4$},\\
%   1, 
%   &\mbox{if  $ m=1$}\\
  1-{\cal P}_{m,l-4}^{\rm Proa}, 
  &\mbox{if  $ l\geqslant 5$ },
  \end{cases}
\end{align}
and for $l \le 4$,
\begin{align}\label{trans_l4_m}
    &{\Theta^{\rm Proa} [{n},{m},l]}=1-\prod\limits_{r = 1}^{l}(1-{\mathbb P}_{m,r})\\\nonumber
    &=\sum\limits_{{r} = 1}^{{l}} {{( - 1)}^{{r} + 1}}\Big( \begin{array}{l}
{l}\\
{r}
\end{array} \Big)
\exp \big( {- \displaystyle\frac{{r\gamma_{\rm th}}\sigma ^2}{g_m{\rho }}\big)}{{{(1 + {{\gamma_{\rm th}}})}^{-kn}}}
\nonumber \\
% \exp \big( {\displaystyle{{ - k{\gamma _{th}}{\sigma ^2}}}/{\rho }\big)}
&\times\displaystyle{{\exp \Big(
%\displaystyle\frac{{ - k{\gamma _{th}}{\sigma ^2}}}{\rho } 
{ - \displaystyle{{{}{{\eta_{m,r}}{\mathcal{A}_m^{Proa}}\lambda _a}}}/{{{\lambda _B}}}\Big( {{}_2{F_1}\Big( { - \frac{2}{\alpha },k;\frac{{\alpha  - 2}}{\alpha }; - {{\gamma_{\rm th}}}} \Big) - 1} \Big)} \Big)}},
\end{align}
% \begin{align}\label{trans_l4_m}
%     &{\Theta^{\rm Proa} [{n},{m},l]}=1-\prod\limits_{r = 1}^{l}(1-{\mathbb P}_{m,r}),
% \end{align}
% with
% \begin{align}\label{mtrans}
%         {\mathbb P}_{m,r}&={\eta_{m,r}}
%         %\sum\limits_{n = 0}^\infty 
%       {{\Theta^{\rm Proa} [{n},{m},1]}}\\\nonumber
%       &={\eta_{m,r}}
% \displaystyle\frac{{\exp \Big({{ - {\gamma _{th}}{\sigma ^2}}}/({g_m\rho})
% %\displaystyle\frac{{ - k{\gamma _{th}}{\sigma ^2}}}{\rho } 
% { - \displaystyle{{{\eta_{m,r}}{\mathcal{A}_m^{Proa}}{\lambda _a}}}/{{{\lambda _B}}}\Big( {{}_2{F_1}\Big( { - \frac{2}{\alpha },1;\frac{{\alpha  - 2}}{\alpha }; - {\gamma _{th}}} \Big) - 1} \Big)} \Big)}}{{{(1 + {\gamma _{th}})}^{n}}},
%     \end{align}
% \begin{align}\label{trans_l4_m}
%     &{\Theta^{\rm Proa} [{n},{m},l]}=1-\prod\limits_{r = 1}^{l}(1-{\mathbb P}_{m,r})
%     \\\nonumber
%     &=\sum\limits_{{r} = 1}^{{l}} {{( - 1)}^{{r} + 1}}\Big( \begin{array}{l}
% {l}\\
% {r}
% \end{array} \Big)
% % \exp \big( {\displaystyle{{ - k{\gamma _{th}}{\sigma ^2}}}/{\rho }\big)}
% \displaystyle\frac{{\exp \Big({{ - r{\gamma _{th}}{\sigma ^2}}}/{\rho }
% %\displaystyle\frac{{ - k{\gamma _{th}}{\sigma ^2}}}{\rho } 
% { - \displaystyle{{{\eta_{m,r}}{\mathcal{A}_m^{Proa}}{\lambda _a}}}/{{{\lambda _B}}}\Big( {{}_2{F_1}\Big( { - \frac{2}{\alpha },r;\frac{{\alpha  - 2}}{\alpha }; - {\gamma _{th}}} \Big) - 1} \Big)} \Big)}}{{{(1 + {\gamma _{th}})}^{rn}}},
% \end{align}
and for $l \ge 5$, 
 \begin{align}\label{trans_l5_m}
{\Theta^{\rm Proa} [{n},{m},l]} =1-(1-{{\Theta^{\rm Proa} [{n},{m},4]}})\prod\limits_{r = 5}^{l}\big({1-\mathbb P}_{m,r}\big),
\end{align}
with 
\begin{align}\label{m=1trans}
        {\mathbb P}_{m,r}={\eta_{m,r}}
        %\sum\limits_{n = 0}^\infty 
        {{{{\rm O}[{n},m,r]}{\Theta^{\rm Proa} [{n},{m},1]}}},
    \end{align}
where ${{\Theta^{\rm Proa} [{n},{m},1]}}$ is obtained from (\ref{trans_l4_m}) and ${{\rm O}[{n},m,r]}$ is obtained from (\ref{N_condition1mu_proa_m}).
\end{lemma}

Finally,  the latent access failure probabilities for the Proactive scheme under an arbitrary latency constraint can be obtained using the iteration process shown in Fig. 8 with the details described in the following.
\begin{figure}
	\centering
	\includegraphics[width=3.5in,height=4.1in]{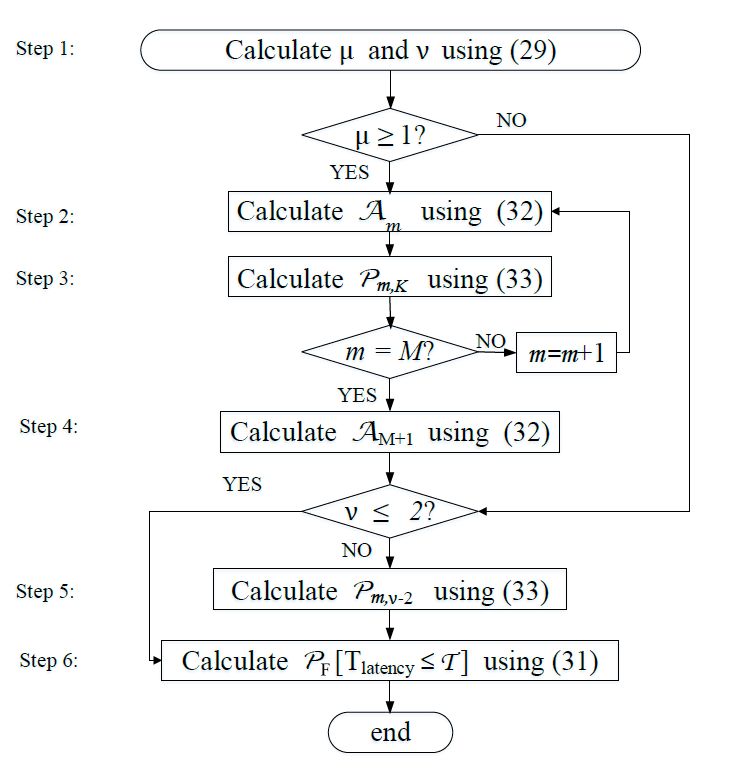}
	\caption{Flowchart for deriving the latent access failure probability of the Proactive scheme.}
	\label{fig:my_label}
\end{figure}

% $\mathcal{A}_m^{\rm Proa}$ and  is given according to \textbf{Theorem 3} with $l=K$ when $1\leqslant m\leqslant M-1$ and $l\leqslant K$ when $m=M$, respectively;
% \begin{align}\label{active probability_r}
% {\mathcal{A}_m^{Pb}} = \bigg\{ \begin{gathered}
%   {\cal P}_{\rm out}^{Pb}[T_{\rm latency}>{\cal T}_{m}], m>1, \hfill \\
%   1, m = 1. \hfill \\ 
% \end{gathered}
% \end{align} 

\begin{itemize}
\item Step 1: Calculate the indexes $\mu$ and $\nu$  under the given latency constraint ${\cal T}_{}$ {\rm TTIs} usinig (\ref{index}).  If $\mu\ge1$,  go to Step 2; If $\mu=0$, $\nu\le 2$ , go to Step 6; If $\mu=0$, $\nu\ge 3$ , go to Step 5;

%  \item Step 2: If $K_{\rm Proa}\le 4$, go to flowchart in Fig. 6, otherwise go to the step 3;
 
 \item Step 2: Calculate non-empty probability ${\mathcal{A}_m^{Proa}}$ using (\ref{active_proa});

  \item Step 3: Calculate the GF access success probability in the $m$th round trip, ${\cal P}_{m,K}^{\rm Proa}$ using (\ref{data_proa_m});

  Repeating Step 2 to 3 until $m=M$;
  
  \item Step 4: Calculate non-empty probability ${\mathcal{A}_{M+1}^{Proa}}$ using (\ref{active_proa});
  
  \item Step 5:  If $\nu\ge 3$,  calculate the GF access success probability ${\cal P}_{m,\nu-2}^{\rm Proa}$   using (\ref{data_proa_m});

    \item Step 6: Calculate the latent access failure probability  ${\cal P}_{\rm out}^{\rm Proa}[T_{\rm latency}\leqslant{\cal T}_{}]$ using (\ref{promulti_m}).

 \end{itemize}

\section{Simulation and Discussion}
In this section, we verify our analytical results by comparing the theoretical GF latent access failure probabilities with the results from Monte-Carlo simulations, where the simulations are performed using the system model described in Section II in MATLAB. 
The BSs and UEs are deployed via independent HPPPs in a 1600 km$^2$ circle area  with each UE associated with its nearest BS.
At the beginning of each round trip, UEs randomly move to new positions and the active ones randomly choose a pilot from $S=48$ pilots to transmit.
The channel fading gains between the UEs and BSs are modeled by exponentially distributed random variables.
The simulation parameters used for this study are in line with the main guidelines for 3GPP NR performance evaluations presented in \cite{2Tel2018} with mini-slots of 7 OFDM symbols for transmissions
in short TTI (0.125ms) using 60 kHz SCS{\footnote{{Mini-slot durations will depend on the SCS and on the
number of OFDM symbols for a given SCS, adopted according to the type of deployment and carrier frequency.
}}}. 
{To focus on the GF access in UL, we assume feedback in DL is error-free.\footnote{{According to Section II.C,   a HARQ  round-trip includes:
1) UL (UE to BS): the UE transmits the signal to the BS and the BS decodes the received signal;
2) DL (BS to UE): the BS sends an ACK/NACK feedback and the UE  processes the feedback to decide whether to perform a retransmission in the next HARQ.
This is to say,  from the UE perspective, one HARQ  round-trip is finished until the UE processes the feedback to know whether it is successful or not, which should consider both the transmission probability in UL and DL.
In this paper, to focus on the GF access in UL, we assume feedback in DL is error-free.
The analysis of the feedback with error probability can be extended following this work.}}}.
% Our results can be easily extended to scenarios with error probability.
% It is assumed that the receiver can ideally estimate the channel of all superimposed% transmissions. 
The simulation time is configured to collect at least $5\times10^6$ samples to ensure a sufficient confidence level on the $10^{−5}$ quantile.
In all figures of this section, “Analytical” and “Simulation” are abbreviated as “Ana.” and “Sim.”, respectively.
% On the vertical axes, the latent access failure probabilities, i.e., ${\cal P}_{F}=1-{\cal P}_{F}^{\rm Proa}[T_{\rm latency}\leqslant{\cal T}_{}]$ are displayed. 
Unless otherwise stated, we consider $\lambda_{\rm{B}}=1$ BSs/km$^2$, $\lambda_D=20000$ $\rm {UEs/km^2}$, ${\gamma_{\rm th}}= -2$ dB, $\alpha = 4$, $\rho =−130$ dBm, $p_a=0.0011$, $g_J=g_1=1$, 
the noise $\sigma^2$ = −174+10log$_{10}$(60000) = −126.2 dBm.

Fig. 9-Fig. 10 plot the GF latent access failure probabilities of the UE with the Reactive, K-repetition, and Proactive schemes  versus SINR thresholds ${\gamma_{\rm th}}=-10 $dB and ${\gamma_{\rm th}}=-2 $dB, respectively.
%  are set as $\gamma_{th}=-10 $dB and $\gamma_{th}=-2 $dB, respectively.and SINR thresholds, respectively.
% Based on the relative values of these parameters, we identify four main network configurations, which are 1) LD-HS: configuration with low density of UEs and high SINR threshold;  2) HD-HS: configuration with high density of UEs and high SINR threshold; 3) LD-LS: configuration with low density of UEs and low SINR threshold; and 4) HD-LS: configuration with high density of UEs and low SINR threshold. 
% On the horizontal axes, the latency constraints are plotted using the unit of TTI (1 TTI = 0.125 ms).
The analytical curves of the Reactive and K-repetition schemes are plotted following the flowchart in Fig. 6, and the analytical curves of the Proactive scheme are plotted following the flowchart in Fig. 8.
The close match between the analytical curves and simulation points validates the accuracy of the developed spatio-temporal mathematical framework.
The stair behaviour (i.e., the latent access failure probabilities stay unchanged for a period of time) is caused by the waiting time between each retransmission.

\begin{figure}[htbp!]
	\begin{center}
		\begin{minipage}[t]{0.52\textwidth}
			\centering
			\includegraphics[width=3.5in,height=2.6in]{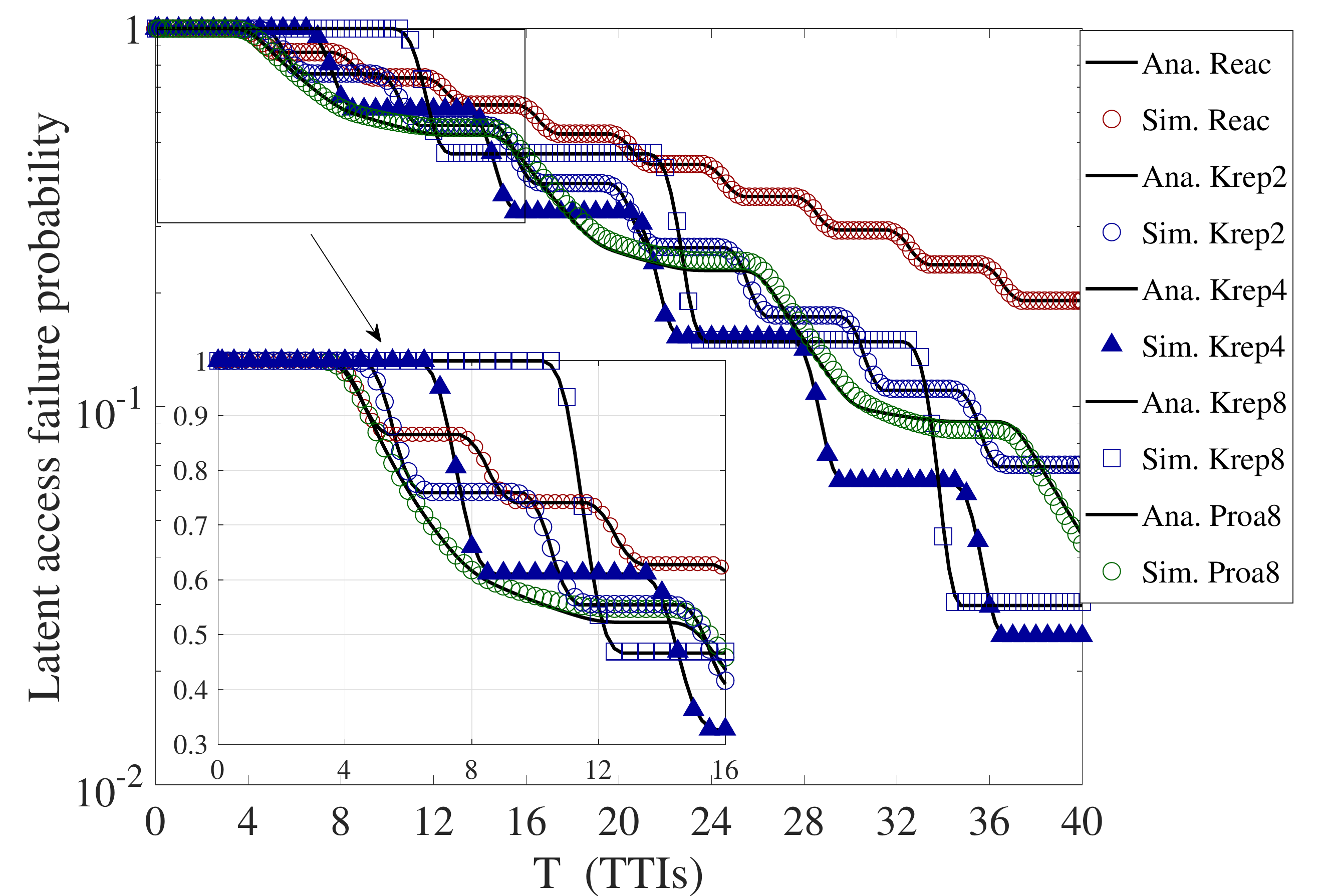}
			\vspace*{-0.4cm}
			\caption{\scriptsize Latent access failure probability when ${\gamma_{\rm th}}=-2$dB.}
		\end{minipage}
		\label{fig:10}
		\begin{minipage}[t]{0.45\textwidth}
			\centering
			\includegraphics[width=3.5in,height=2.6in]{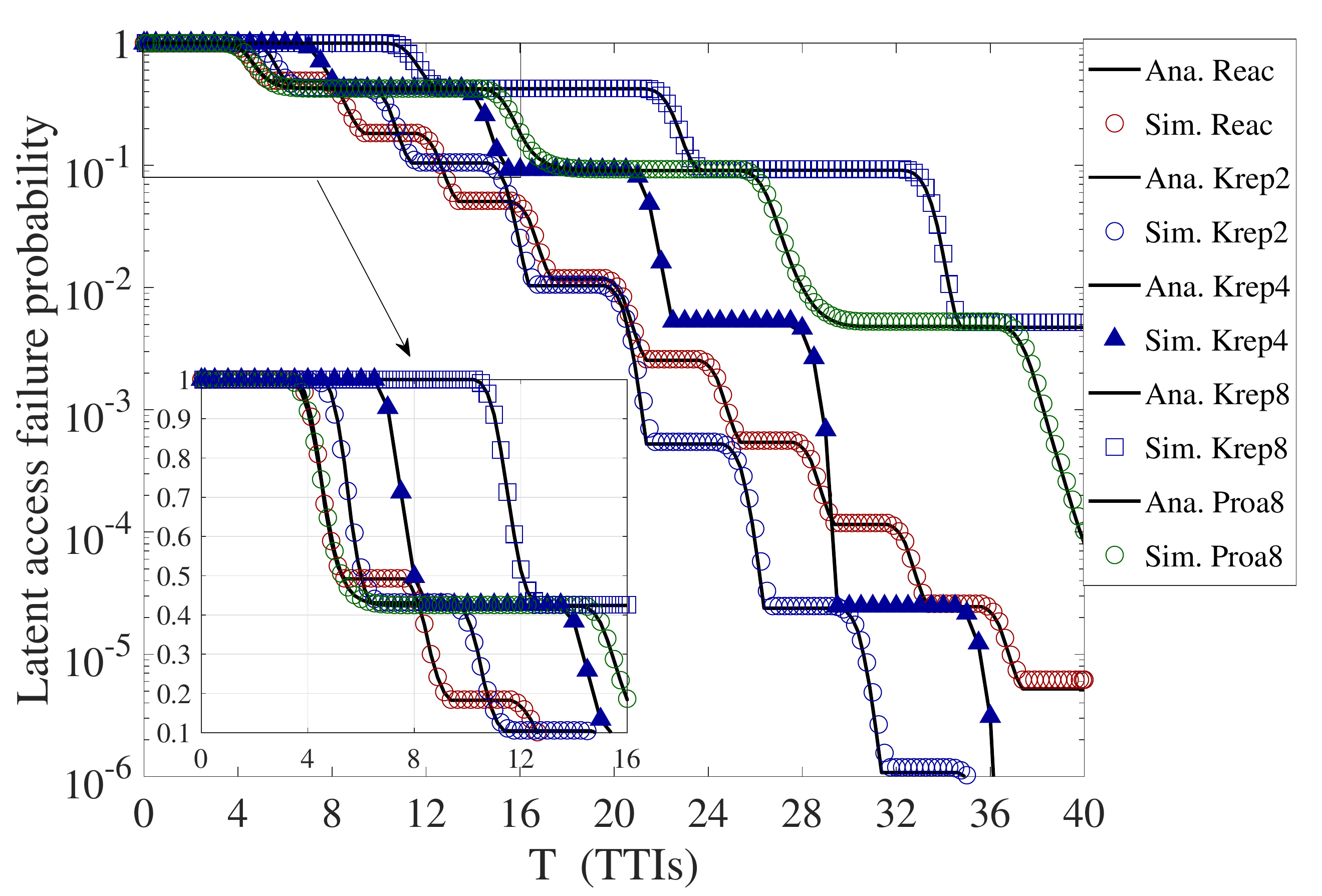}
			\vspace*{-0.4cm}
			\caption{\scriptsize Latent access failure probability when ${\gamma_{\rm th}}=-10$ dB.}
		\end{minipage}
		\label{fig:11}
	\end{center}
\end{figure}

In Fig. 9, {we first observe that the latent access failure probabilities follow  Proa$=$Reac$\le$Krep under latency constraints ${\cal T}\le 0.625$ms (5TTIs).
%we first observe that the latent access failure probabilities of the K-repetition scheme with different repetition values are lower than those of the Reactive scheme under longer latency constraints ${\cal T}\ge 1.5$ms (12 TTIs).
%This is due to that there is not enough time to perform K$\ge$2 repetitions.
}
{We also observe that the latent access failure probabilities follow  Proa$\le$Krep$\le$Reac under latency constraints $0.625$ ms $\le{\cal T}\le 1.5$ms (12TTIs).
%we first observe that the latent access failure probabilities of the K-repetition scheme with different repetition values are lower than those of the Reactive scheme under longer latency constraints ${\cal T}\ge 1.5$ms (12 TTIs).
In this case, under shorter latency constraints ${\cal T}\le 12$TTIs, the Proactive scheme should be chosen.
}
{This is due to that the Proactive scheme could terminate earlier to reduce latency without waiting for $K$ repetitions, which satisfies the shorter latency constraints}.
% We also observe that the latent access failure probabilities of the 4-repetition scheme are lower than those of  the 2-repetition scheme after longer latency constraints ${\cal T}\ge 1.8$ms (15 TTIs).
{But when the latency constraints ${\cal T}$ get longer, the advantage of the Proactive scheme than the K-repetition scheme is not obvious but the advantage of the Proactive and K-repetition schemes than the Reactive scheme is obvious, i.e., Proa$\&$Krep$<$Reac,  due to that the UE has enough time to finish the repetitions and get feedback.}
{We note that} increasing repetition value increases the GF access success probability, as it offers more opportunities to retransmit. 
However, when the repetition value is too large (e.g., $K_{\rm Krep}=8$), the latent access failure probabilities are not lower than those of the 4-repetition scheme in most of the time (except $1.5 \rm{ms}\le{\cal T}\le 1.8 \rm{ms}$, $4.2 \rm{ms}\le{\cal T}\le 4.5\rm{ms}$).
This is due to that transmitting 8 repetitions will cost too much waiting time and introduce a much longer delay.
It is obvious that if the repetition value is overestimated, the K-repetition scheme will waste the potential resource and lead to lower resource efficiency.

% However, it does not always perform better due to that the Proactive scheme also help to reduce interference to other UE(s) to improve reliability.
% This can be explained from the fact that the Proactive scheme earliest determination time depends on the HARQ RTT. 
% This can be explained from the fact that the Proactive scheme earliest determination time depends on the HARQ RTT. 
% Since in the HD-LS scenario, more than 2 repetitions are rarely needed, the Proactive scheme with max 8 repetitions is not better than other schemes, except for the K-repetition scheme with 8 repetitions, in this case.

% According to the curve of the Proactive scheme with maximum 8 repetitions, we can see that the blind repetition presents worse SINR due to the extra interference caused by the blind repetitions.

% This is due to that the UE queueing delay for $K=8$ is too long.
% GF with the K-repetition on the other hand presents the worst SINR due to the extra intra-cell interference caused by the blind repetitions. 
% The GF Reactive scheme presents a better SINR then the other GF schemes given that it avoids unnecessary retransmissions. 
% This shows that each transmission of the
% Reactive scheme presents a higher reliability, compared to the cases with blind repetitions. 

%repetitions will not improve the GF RA performance any more. 
% This is consistent with Fig. 4 as there is a  trade off  between  transmission  success  probability  and  non-collision  probability.

In  Fig. 10, 
{we first observe that latent access failure probabilities follow Krep$\le$Reac$\le$Proa, under longer latency constraints ${\cal T}\ge 1.5$ms (12 TTIs) for small repetition value $K_{\rm Krep}=2$.
%we observe that the 2-repetition scheme outperforms the Reactive scheme after 2 HARQ round trips, that the 4-repetition scheme outperforms the Reactive scheme after 4 HARQ round trips, and that the 8-repetition scheme has the highest latent access failure probabilities.
In this case, under longer latency constraints ${\cal T}\ge 12$TTIs, the K-repetition scheme should be chosen.
}
{
We also observe that the 8-repetition scheme has the highest latent access failure probabilities.}
This is due to that there is a  trade-off between transmission success probability and non-collision probability when increasing the repetition value, which is in line with Fig. 5 (b).
% In Fig. 5 (b), we note that in the scenario the same as  Fig. 10 ($\gamma_{th}=-10$ dB and ${\lambda_D}=2\times10^4$ $\rm {UEs/km^2}$), increasing the repetition value (from 2 to 4) does not increase the GF access success probability.
Thus, in Fig. 10, increasing the repetition value to 8 does not decrease but increases the latent access failure probabilities because it introduces longer waiting time without increasing the access success probabilities.
{In this case, when the repetition value is overestimated, the Proactive scheme should be chosen.}
% Similarly, the latent access failure probabilities of the Proactive scheme with a maximum of 8 repetitions could only outperform those of the K-repetition scheme with the same repetition value but underperform those of other schemes.
% Thus, it is necessary to deduce the  optimal number of repetitions in the future work.

% However, since UEs usually tend to operate using maximum transmit power, it is not able to apply Power Boosting in many cases.
% In addition, too much energy consumption in the Power boost scheme leads it not to be a regular scheme in the energy saving scenarios.
\begin{figure}[htbp!]
    \begin{center}
    \begin{minipage}[t]{0.53\textwidth}
    \centering
        \includegraphics[width=3.5in,height=3in]{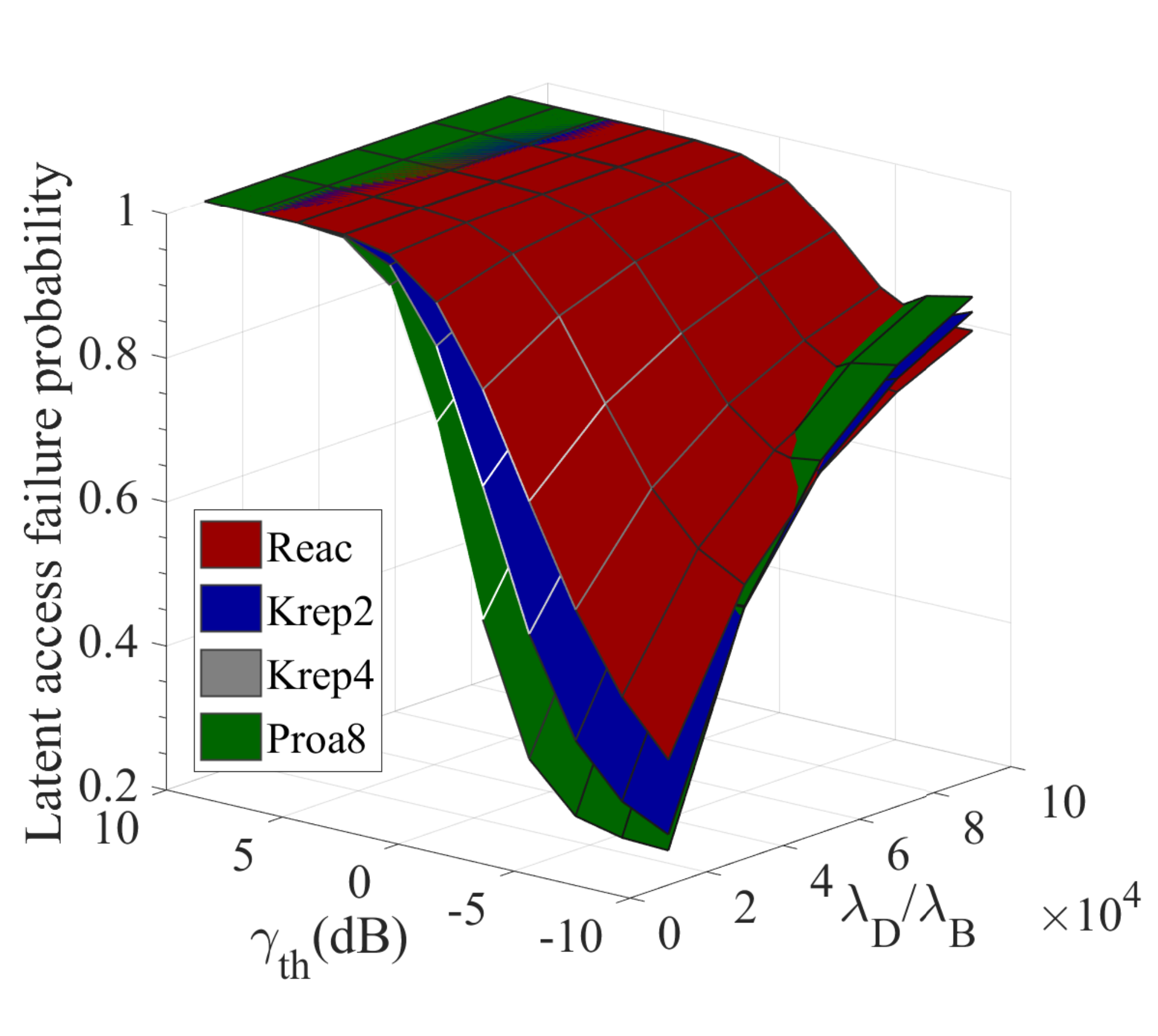}
        \vspace*{-0.4cm}
        \caption{\scriptsize Latent access failure probability for different density ratios and SINR thresholds when ${\cal T}=8$ TTIs (1ms).}
    \end{minipage}
    \label{fig:10}
        \begin{minipage}[t]{0.45\textwidth}
    \centering
        \includegraphics[width=3.5in,height=3in]{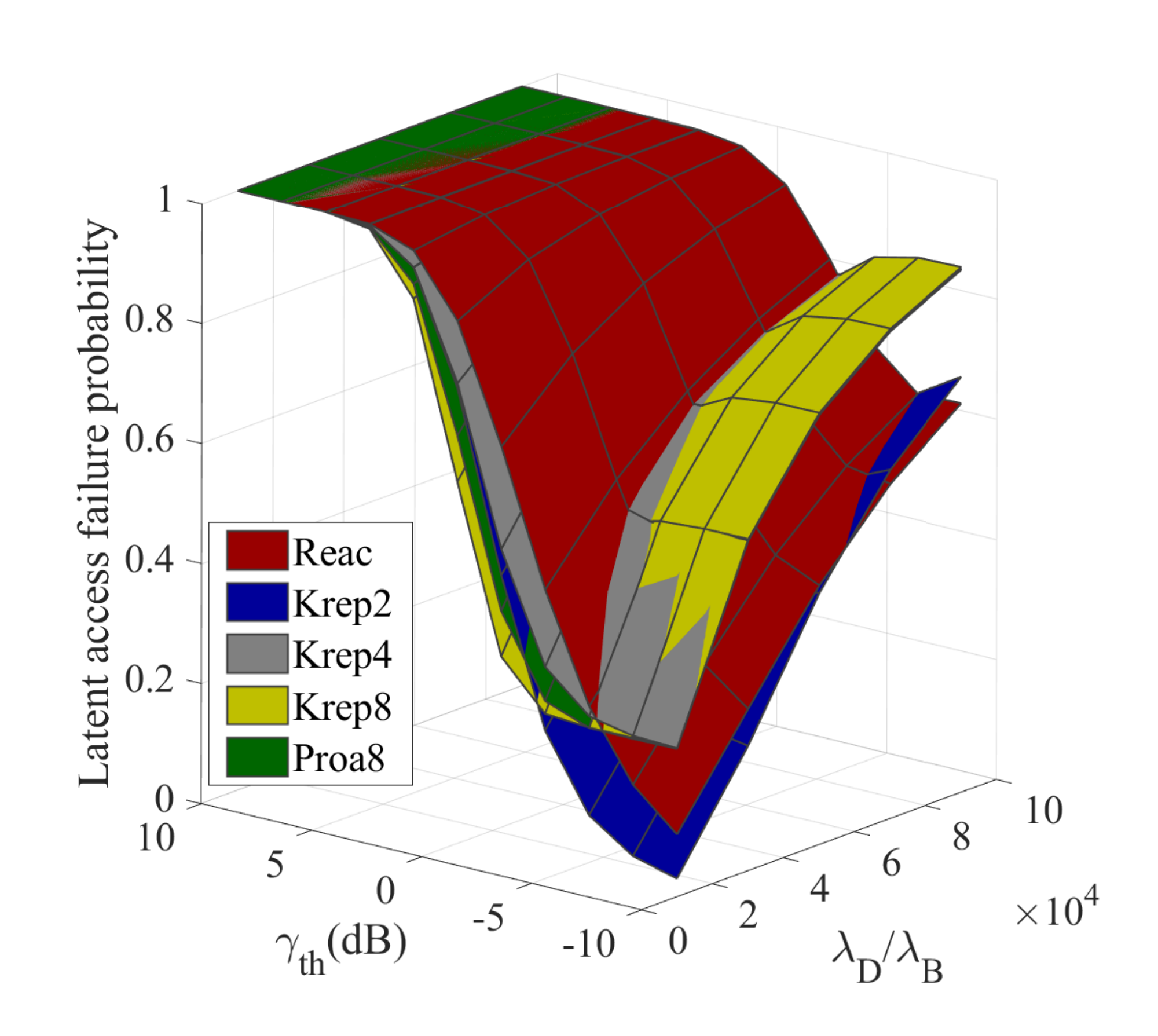}
        \vspace*{-0.4cm}
        \caption{\scriptsize Latent access failure probability for different densities ratios and SINR thresholds when ${\cal T}=12$ TTIs (1.5ms).}
        \end{minipage}
        \label{fig:11}
    \end{center}
\end{figure}

Fig. 11-Fig. 12 plot the GF latent access failure probabilities under the latency constraint ${\cal T}=1$ms  (8 TTIs) and  ${\cal T}=1.5$ms  (12 TTIs) for different density ratios and SINR thresholds. 
We observe that the GF latent access failure probability increases with increasing density ratio which is due to the following two reasons: 1) increasing the number of UEs generating interference leads to lower received SINR at the BS; 2) increasing the number of UEs leads to higher probability of collision. 
We also observe that the GF latent access failure probabilities decreases with decreasing SINR  threshold.
This is due to the lower SINR threshold leading to higher access success probability.

In Fig. 11, we observe that the GF latent access failure probabilities decrease in light load scenario (e.g., $\lambda_D/\lambda_B\le 40000$), while increases in high load scenario (e.g., $\lambda_D/\lambda_B\ge 40000$) with increasing the repetition value, which is in line with Fig. 5 (b).
This is due to the fact that increasing the repetition increases the collisions in overloaded traffic scenario, and wastes extra time and frequency resource.
% Thus, despite the potential of K-repetition to achieve a lower outage, it is not suited for this scenario.
We also note that, as the latency constraint ${\cal T}=1$ms (8 TTIs), so the 8-repetition scheme can not be adopted because its waiting time for the 1st transmission is more than 1ms.
But the Proactive scheme with a maximum of 8 repetitions could have as good performance as the 4-repetition scheme.

In Fig. 12, we observe that the GF latent access failure probabilities decrease in higher SINR thresholds scenarios (e.g., ${\gamma_{\rm th}}\ge -5$dB), while increases in lower SINR thresholds scenarios (e.g., ${\gamma_{\rm th}}\le -5$dB) with increasing the repetition value $K_{\rm Krep}>2$.
Thus, despite the K-repetition scheme can cope with tight time constraints by allowing a number of consecutive repetitions in a short time, the interference due to the multiple repetitions is the major impacting factor and surpasses the benefits of the combining gain in lower SINR threshold and high density scenarios.

% Fig. 14 plots the GF latent access failure probabilities under the latency constraint ${\cal T}=8$ TTIs (1ms) for different density ratios and SINR thresholds. 
% The K-repetition scheme always outperforms the Reactive scheme, and the performance gets better with increasing the repetition value,
% but the gap decreases when $\gamma_{th}\ge -5$dB.
% In addition, for the short latency constraint ${\cal T}=1$ms (8 TTIs), K-repetition schemes are not available in high load and low SINR threshold scenario. 

\begin{figure}[htbp!]
    \begin{center}
    \begin{minipage}[t]{0.53\textwidth}
    \centering
        \includegraphics[width=3.4in,height=2.8in]{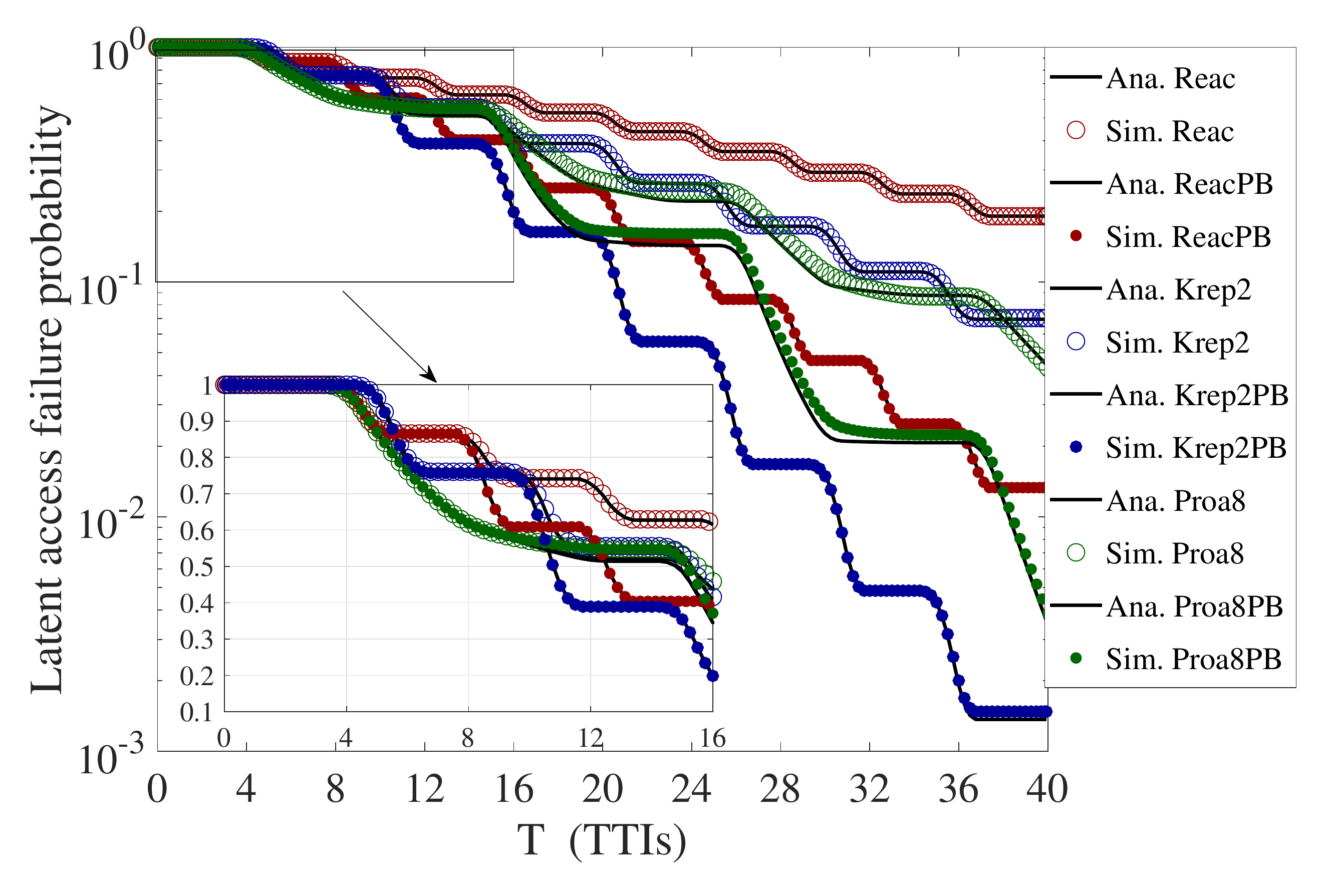}
        \vspace*{-0.4cm}
        \caption{\scriptsize Latent access failure probability when ${\gamma_{\rm th}}=-2$dB.}
    \end{minipage}
    \label{fig:10}
        \begin{minipage}[t]{0.45\textwidth}
    \centering
        \includegraphics[width=3.7in,height=2.8in]{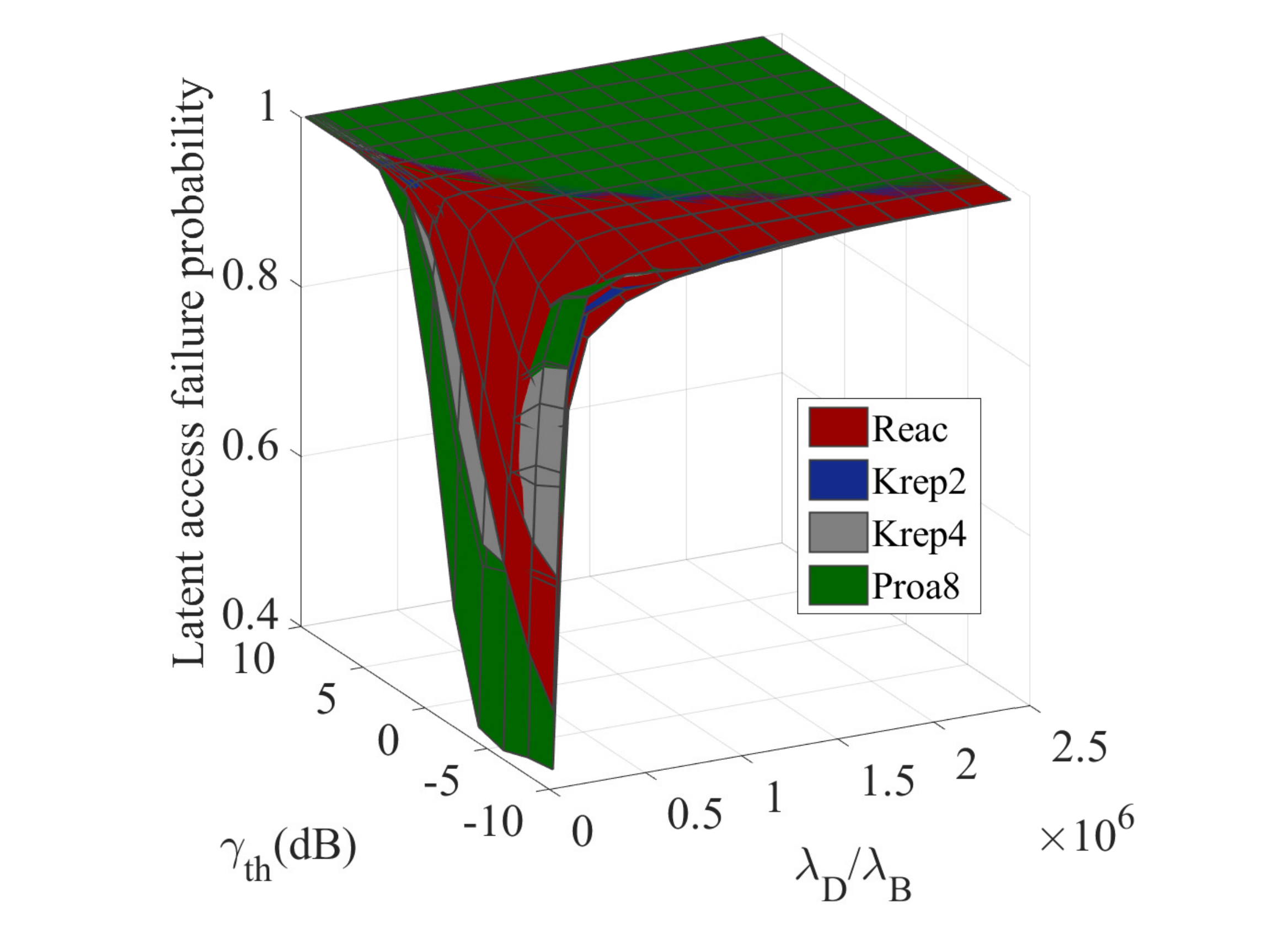}
        \vspace*{-0.4cm}
        \caption{\scriptsize Latent access failure probability when ${\cal T}=8$ TTIs (1ms) for large densities.}
        \end{minipage}
        \label{fig:11}
    \end{center}
\end{figure}

%Fig. 13 plots the GF latent access failure probabilities of the UE under  the Reactive, K-repetition, and Proactive GF HARQ schemes with Power Boosting. We observe that the GF latent access failure probabilities of the GF HARQ schemes with Power Boosting outperform those of the schemes without Power Boosting. This is due to that the failed UEs are favored by stepping up the transmit power, which significantly increases the transmission success probability.Interestingly, we also observe that Power Boosting has a greater improvement on the K-repetition scheme than the other two schemes.

{Fig. 13 plots the GF latent access failure probabilities of the UE under the Reactive, K-repetition, and Proactive schemes with PB.
Interestingly, we observe that PB has greater improvement in the Reactive and K-repetition schemes than the Proactive scheme. For example, without PB, GF latent access failure probabilities of the UE under the K-repetition scheme are similar to those of the Proactive scheme, while with PB, the GF latent access failure probabilities of the UE under the K-repetition scheme is much lower than those of the Proactive scheme.
}

{Fig. 14 plots the GF latent access failure probabilities under the latency constraint ${\cal T}=1$ms  (8 TTIs) versus density ratios and SINR thresholds for larger UE densities. 
We observe that when the density ratios (UE densities) are particularly large ($\lambda_D/\lambda_B\ge 1.5\times 10^6$), no matter what schemes are taken, the GF access cannot be successful, that is, the network is very crowded. Thus, the number of active UEs that access to the network should be limited to some thresholds.}

\section{Conclusion}
In this paper, we developed a spatio-temporal mathematical model to analyze and compare the grant-free access latent access failure probabilities of a randomly chosen UE with three different GF HARQ schemes for URLLC requirements. 
% We first gave an overview of contention-based grant-free random access schemes in 5G NR.
We defined the latent access failure probability to characterize the URLLC performance. 
We proposed a tractable approach to derive and analyze the GF latent access failure probabilities of the UE under the Reactive, the K-repetition, and the Proactive  schemes, respectively.
Our results have shown that
{
1) either K-repetition scheme or Proactive scheme provides lower latent access failure probability than the Reactive scheme, except higher density and lower SINR threshold scenarios; 
2) under shorter latency constraints (${\cal T} \le 8$TTIs), the Proactive scheme provides the lowest latent access failure probability;
%the latent access failure probabilities depend on the latency constraints ${\cal T}$, i.e., which scheme performs better depends on the latency constraints ${\cal T}$;
3) under longer latency constraints, 
%${\cal T} \ge 22$TTIs, 
the K-repetition scheme provides the lowest latent access failure probability, which  depends on $K$, i.e., $K$ need to be optimized; 
{4) if $K$ is overestimated; the Proactive scheme provides lower latent access failure probability than the K-repetition scheme;} 5) the Power Boosting can improve the latent access failure probability, especially for the K-repetition scheme {(including $K$=1)}.}
The analytical model presented in this paper can also be applied for the reliability and latency performance evaluation of other types
of GF HARQ schemes in the cellular-based networks.
% Currently, a lot of 3GPP study items are also focused on enhancing the GF access schemes.
% The analysis in this paper can also be extended with the adoption of other schemes, such as the Non-Orthogonal Coded Access scheme and Hybrid  Resource Allocation scheme, in order to improve the performance further.

\appendices
\numberwithin{equation}{section}
\section{A Proof of  Theorem 1}
% For ease of presentation, we omit the subscript $``{\rm Reac}"$.
For a given latency constraint $\cal T$ {\rm TTIs}, we have $M=\lfloor {({\cal T}-1)}/T^{\rm RTT}_{\rm Reac} \rfloor$.
For $M=1$, the latent access failure probability is the probability that the UE fails to access in the 1st HARQ round trip, where we can derive
\begin{align}\label{latency1}
&{\cal P}_{F}^{\rm }[T_{\rm latency}\leqslant{\cal T}_{}] %\\\nonumber
=1-{{\cal P}_1^{\rm }}. 
% (M=\lfloor {({\cal T}-1)}/T^{\rm RTT}_{\rm Reac} \rfloor=1).     
\end{align}
For $M=2$, the latent access failure probability is the probability that the {UE fails to access in neither two HARQ round trips,} where we can derive 
\begin{align}\label{latency2}
&{\cal P}_{F}^{\rm }[T_{\rm latency}\leqslant{\cal T}_{}] %\\\nonumber
=1-{{\cal P}_1^{\rm }-(1-{\cal P}_1^{}){\cal P}_2^{}}.     
\end{align}
Substituting   (\ref{active_probability})  into   (\ref{latency2}), we have
\begin{align}\label{latency22}
&{\cal P}_{F}^{}[T_{\rm latency}\leqslant{\cal T}_{}] 
%\\ \nonumber
% =1-{{\cal P}_1^{}-\mathcal{A}_2^{}{\cal P}_2^{}}
%\\ \nonumber
=1-\sum\limits_{m = 1}^{M = 2}{\mathcal{A}_m^{}}{{{\cal P}_m^{}}}.
\end{align}
For $M=3$, the latent access failure probability means the probability that the UE {fails to access after all the three HARQ round trips}. So we can derive
\begin{align}\label{latency3}
&{\cal P}_{F}^{}[T_{\rm latency}\leqslant{\cal T}_{}] \nonumber\\
% &=1-{{\cal P}_1^{}-(1-{\cal P}_1^{}){\cal P}_2^{}}- (1-{\cal P}_1^{})(1-{\cal P}_2){\cal P}_3 \\\nonumber
&=1-{{\cal P}_1^{}-(1-{\cal P}_1^{}){\cal P}_2^{}}- (1-{\cal P}_1^{}-(1-{\cal P}_1^{}){\cal P}_2){\cal P}_3\nonumber\\
&=1-{\cal P}_1^{}-\mathcal{A}_2^{}{\cal P}_2^{}-\mathcal{A}_3^{}{\cal P}_3^{}
=1-\sum\limits_{m = 1}^{M = 3}{\mathcal{A}_m^{}}{{{\cal P}_m^{}}}.
\end{align}

For $M>3$, the latent access failure probability ${\cal P}_{F}^{}[T_{\rm latency}\leqslant{\cal T}_{}]$  can be derived based on the iteration process following  $M=2$ and $3$.

\section{A Proof of  Lemma 1}
We derive the GF transmission success probability conditioning on $n$ number of intra-cell interfering UEs based on the SINR outage as 
\begin{align}\label{PROOF_Reac}
 &\Theta^{\rm Reac} [{n},{}m] = { { {{\mathbb P}}[{\rm SINR}_m\geqslant{\gamma_{\rm th}}| {N = n}]} } \nonumber\\
 &=\mathbb P \bigg\{ {\frac{{{}g_m\rho {h_0}}}{{{\mathcal I}_{\operatorname{{\rm inter}}}^m + {\mathcal I}_{\operatorname{{\rm intra}}}^m + {\sigma ^2}}} \geqslant {{\gamma_{\rm th}}} \Big| {N = n} } \bigg\}  \\ \nonumber
   &= \exp \Big( { - \frac{{{{\gamma_{\rm th}}}}}{{{}g_m\rho }}{\sigma ^2}} \Big){{\mathcal L}_{{\mathcal I}_{\operatorname{{\rm inter}}}^m}}\Big( {\frac{{{{\gamma_{\rm th}}}}}{{{g_m}\rho }}} \Big){{\mathcal L}_{{\mathcal I}_{\operatorname{{\rm intra}}}^m}}\Big( {\frac{{{{\gamma_{\rm th}}}}}{{{g_m}\rho }}\Big| {N = n}} \Big).  
\end{align}

The Laplace Transform of aggregate intra-cell interference  conditioning on $N = n$ is derived as
\begin{align}\label{intra_reac}
  {{\mathcal L}_{{\mathcal I}_{\operatorname{{\rm intra}}}^m}}( {s| {N = n} } ) = { E}\Big[ {\exp \Big( { - s\sum\limits_{\beta = 1}^n {{g_m}\rho {h_\beta}} } \Big)} \Big] \hfill  = {\Big( {\frac{1}{{1 + s{g_m}\rho }}} \Big)^n}, 
\end{align}
where $s= { {{{{\gamma_{\rm th}}}}}/({{{g_m}\rho }})}$.

The Laplace Transform of aggregate inter-cell interference received at the BS is derived as
\begin{align}\label{inter_reac}
 & {{\mathcal L}_{{\mathcal I}_{\operatorname{{\rm inter}}}^m}}( s ) = { E}\Big[ {\exp \Big( { - s\sum\limits_{i \in {{\mathcal Z}_{\operatorname{\rm inter}}}} {{g_m}{P_i}{h_i}{{\| {{u_i}} \|}^{ - \alpha }}} } \Big)} \Big] 
   \overset{\text{(a)}} \nonumber\\
   &= { E}\Big[ {\prod\limits_{i \in {{\mathcal Z}_{\operatorname{\rm inter}}}}^{} {\frac{1}{{1 + s{g_m}{P_i}y_i^{ - \alpha }}}} } \Big] \hfill \nonumber\\
  & \overset{\text{(b)}}= \exp \Big( { - 2\pi {\mathcal A}_m^{\rm Reac}{\lambda _a}\int_{{{( {\frac{P}{{{g_m}\rho }}} )}^{\frac{1}{\alpha }}}}^\infty  {{{ E}_P}\Big[ {1 - \frac{1}{{1 + s{g_m}P{y^{ - \alpha }}}}} \Big]} ydy} \Big) \hfill \nonumber\\
  & \overset{\text{(c)}}= \exp \Big( { - 2\pi {\mathcal A}_m^{\rm Reac}{\lambda _a}{{({g_m}s)}^{\frac{2}{\alpha }}}{{ E}_P}[ {{P^{\frac{2}{\alpha }}}} ]\int_{{{( {s{g_m}\rho } )}^{\frac{{ - 1}}{\alpha }}}}^\infty  {\frac{x}{{1 + {x^\alpha }}}} dx}\Big )\nonumber\\
  & \overset{\text{ }}= \exp \Big( { - 2 {\mathcal A}_m^{\rm Reac}{\lambda _a}/{{{\lambda _B}}}{{({}{\gamma_{\rm th}})}^{\frac{2}{\alpha }}}\int_{{{( {{\gamma_{\rm th}} } )}^{\frac{{ - 1}}{\alpha }}}}^\infty  {\frac{x}{{1 + {x^\alpha }}}} dx}\Big ),
%   & \overset{\text{(d)}}= \exp \Big( { -\pi {\mathcal A}_m^{Reac}{\lambda _a}/{{{\lambda _B}}}{{({}\gamma_{th})}^{\frac{2}{\alpha }}}\arctan({{({}\gamma_{th})}^{\frac{1}{2 }}}) }\Big )\\\nonumber
\end{align}
where (a) is obtained by taking the average with respect to $h_i$, (b) follows from the probability generation functional (PGFL) of the PPP, 
(c) follows by changing the variables $x=y/(sP)^{\frac{1}{\alpha}}$
and $E_P[ {{P^{\frac{2}{\alpha }}}}]=\rho^{\frac{2}{\alpha}}/(\pi\lambda_B)$ is the
moments of the transmit power.
Substituting Eq. (\ref{intra_reac}) and Eq. (\ref{inter_reac}) into (\ref{PROOF_Reac}), we derive the transmission success probability in the $m$th round trip as
\begin{align}\label{mth_reac}
&\Theta^{\rm Reac} [{n},{}m]=
\exp \big( {- \displaystyle\frac{{\gamma_{\rm th}}\sigma ^2}{g_m{\rho }}\big)}{{{(1 + {{\gamma_{\rm th}}})}^{-n}}}
\nonumber \\
&
\times\displaystyle{{\exp \Big(
%\displaystyle\frac{{ - k{\gamma _{th}}{\sigma ^2}}}{\rho } 
{ - \displaystyle{{{}{\mathcal{A}_m^{\rm Reac}\lambda _a}}}/{{{\lambda _B}}}\Big( {{}_2{F_1}\Big( { - \frac{2}{\alpha },1;\frac{{\alpha  - 2}}{\alpha }; - {{\gamma_{\rm th}}}} \Big) - 1} \Big)} \Big)}}.
\end{align}
We consider a general fading with the path loss exponent $\alpha=4$ to simplify our results as
\begin{align}\label{mth_reac_sim}
&\Theta^{\rm Reac} [{n},{}m]=
\exp \big( {- \displaystyle\frac{{\gamma_{\rm th}}\sigma ^2}{g_m{\rho }}\big)}{{{(1 + {{\gamma_{\rm th}}})}^{-n}}}
\nonumber \\
&\times\displaystyle{{\exp \Big(
%\displaystyle\frac{{ - k{\gamma _{th}}{\sigma ^2}}}{\rho } 
{ - {{( {{{\gamma_{\rm th}}}} )}^{\frac{1}{2 }}}\displaystyle{{{}{\mathcal{A}_m^{\rm Reac}\lambda _a}}}/{\lambda _B}{{}}\arctan({{({}{\gamma_{\rm th}})}^{\frac{1}{2 }}}) }\Big)}}.
\end{align}

\section{A Proof of  Lemma 2}
For the K-repetition scheme, the GF transmission in one HARQ round trip is successful if any of the repetition succeeds. We derive the GF transmission success probability under $K_{\rm Krep}$ repetitions conditioning on $n$ number of intra-cell interfering UEs based on the SINR outage as 
\begin{align}\label{PROOF_TRANS}
 &{\Theta^{\rm Krep}} [{n},{m},K_{\rm Krep}] \nonumber \\
& = 1 - \prod\limits_{{k} = 1}^{{K_{\rm Krep}}} {\Big( {1 - {{\mathbb P}}[{\rm SINR}_{k}^m\geqslant{\gamma_{\rm th}}| {N = n}]} \Big)}. 
\end{align}
Based on the Binomial theorem,  \eqref{PROOF_TRANS} can be rewritten as
\begin{align}\label{krep_access}
&{\Theta^{\rm Krep}} [{n},{m},K_{\rm Krep}]= {\sum\limits_{{{k}} = 1}^{{{K_{\rm Krep}}}}}{{( - 1)}^{{{k}} + 1}} {\Big( \begin{array}{l}{{K_{\rm Krep}}}\\{{k}}\end{array}  \Big)}\nonumber \\
&\times{{\mathbb{P}}}[{{{\rm SINR}_1^m}\geqslant{\gamma_{\rm th}}},... ,{{{\rm SINR}_{{{k}}}^m}\geqslant{\gamma_{\rm th}}}| {N = n} ],
\end{align}
where $\Big( \begin{array}{l}
{K_{\rm Krep}}\\
{k}
\end{array} \Big) =\displaystyle \frac{{{K_{\rm Krep}}!}}{{{k}!\big( {{K_{\rm Krep}} - {k}} \big)!}}$ is the binomial coefficient and
% $\displaystyle{{\mathbb P}[{{\rm SINR}_1^m \geqslant {\gamma _{th}},...,{\rm SINR}_k^m\geqslant {\gamma _{th}}| {N = n}} ]}$ is the probability that all of $k$ transmissions are successful derived as
\begin{align}\label{A1}
&{{\mathbb P}[{{\rm SINR}_1^m \geqslant {{\gamma_{\rm th}}},...,{\rm SINR}_k^m\geqslant {{\gamma_{\rm th}}}| {N = n}} ]}\nonumber \\ 
 &= \exp \Big( { - \frac{{{k{\gamma_{\rm th}}}}}{{{g_m}\rho }}{\sigma ^2}} \Big){{\mathcal L}_{I_{\operatorname{{\rm inter}}}^m}}\Big( {\frac{{{{\gamma_{\rm th}}}}}{{{g_m}\rho }}} \Big){{\mathcal L}_{I_{\operatorname{{\rm intra}}}^m}}\Big( {\frac{{{{\gamma_{\rm th}}}}}{{{g_m}\rho }}\Big| {N = n}} \Big).
\end{align}
The Laplace Transform of aggregate intra-cell interference conditioning on $N = n$ is derived as
\begin{align}\label{intra}
  &{{\mathcal L}_{I_{\operatorname{{\rm intra}}}^m}}( {s| {N = n} } )  
  \nonumber \\
  &={ E}\Big[ {\exp \Big( { - s\sum\limits_{\beta = 1}^n {{g_m}\rho \sum\limits_{r = 1}^k {h_\beta^r}} } \Big)} \Big] \hfill  = {\Big( {\frac{1}{{1 + s{g_m}\rho }}} \Big)^{kn}}, 
\end{align}
where $s= { {{{{\gamma_{\rm th}}}}}/({{{g_m}\rho }})}$.

The Laplace Transform of aggregate inter-cell interference  is derived as
\begin{align}\label{inter}
 &{{\mathcal L}_{I_{\operatorname{{\rm inter}}}^m}}( s ) \nonumber \\
 &= { E}\Big[ {\exp \Big( { - s\sum\limits_{i \in {{\mathcal Z}_{\operatorname{\rm inter}}}} {{g_m}{P_i}({\sum\limits_{r = 1}^k h_i^r}){{\| {{u_i}} \|}^{ - \alpha }}} } \Big)} \Big] 
  \overset{\text{}}\nonumber \\
  &= { E}\Big[ {\prod\limits_{i \in {{\mathcal Z}_{\operatorname{\rm inter}}}}^{} \Big({\frac{1}{{1 + s{g_m}{P_i}y_i^{ - \alpha }}}}\Big)^k } \Big] \hfill \nonumber \\
  & \overset{\text{}}= \exp \Big( { - 2\pi {\mathcal A}_m^{Krep}{\lambda _a}\int_{{{( {\frac{P}{{{}\rho }}} )}^{\frac{1}{\alpha }}}}^\infty  {{{ E}_P}\Big[ {1 - (\frac{1}{{1 + s{g_m}P{y^{ - \alpha }}}})^k} \Big]} ydy} \Big) \hfill \nonumber \\
%   & \overset{\text{()}}= \exp \Big( { - 2\pi {\mathcal A}_m^{Krep}{\lambda _a}{{({}s)}^{\frac{2}{\alpha }}}{{ E}_P}[ {{P^{\frac{2}{\alpha }}}} ]\int_{{{( {s{}\rho } )}^{\frac{{ - 1}}{\alpha }}}}^\infty  {\Big[ {1 - (\frac{1}{{1 + {}{x^{ - \alpha }}}})^k} \Big]} xdx}\Big )\\\nonumber
   & \overset{\text{}}= \exp \Big( { - 2 {\mathcal A}_m^{Krep}\frac{\lambda _a}{\lambda _B}({{\gamma_{\rm th}}})^{\frac{2}{\alpha }}\int_{{{( {{\gamma_{\rm th}}} )}^{\frac{{ - 1}}{\alpha }}}}^\infty  {\Big[ {1 - (\frac{1}{{1 + {}{x^{ - \alpha }}}})^k} \Big]} xdx}\Big ).
\end{align}
% where (a) is obtained by taking the average with respect to $h_i$, (b) follows from the probability generation functional (PGFL) of the PPP, 
% (c) follows by changing the variables $x=y/(sP)^{\frac{1}{\alpha}}$
% and $E_P[ {{P^{\frac{2}{\alpha }}}}]=\rho^{\frac{2}{\alpha}}/(\pi\lambda_B)$ is the
% moments of the transmit power. 
Substituting (\ref{intra}) and  (\ref{inter}) into (\ref{A1}) and then substituting (\ref{A1}) into (\ref{krep_access}), we derive the transmission success probability in the $m$th roudn trip with the K-repetition scheme as
\begin{align}\label{trans_krep_proof}
&\Theta^{\rm Krep} [{n},{m},K_{\rm Krep}]
\nonumber \\
&=\sum\limits_{{k} = 1}^{{K_{\rm Krep}}} {{( - 1)}^{{k} + 1}}\Big( \begin{array}{l}
{K_{\rm Krep}}\\
{k}
\end{array} \Big)
\exp \big( {- \displaystyle\frac{{k\gamma_{\rm th}}\sigma ^2}{g_m{\rho }}\big)}{{{(1 + {{\gamma_{\rm th}}})}^{-kn}}}
\nonumber \\
&\times\displaystyle{{\exp \Big(
%\displaystyle\frac{{ - k{\gamma _{th}}{\sigma ^2}}}{\rho } 
{ - \displaystyle{{{}{\mathcal{A}_m^{\rm Krep}\lambda _a}}}/{{{\lambda _B}}}\Big( {{}_2{F_1}\Big( { - \frac{2}{\alpha },k;\frac{{\alpha  - 2}}{\alpha }; - {{\gamma_{\rm th}}}} \Big) - 1} \Big)} \Big)}}.
% &\displaystyle\frac{{\exp \bigg( {\displaystyle\frac{{ - r{\gamma _{th}}{\sigma ^2}}}{\rho } - 2{{( {{\gamma _{th}}} )}^{\frac{2}{\alpha }}}\displaystyle\frac{{{}{\mathcal{A}_m^{k_{rep}}\lambda _a}}}{{{\lambda _B}}}\int_{{{( {{\gamma _{th}}} )}^{ - \frac{1}{\alpha }}}}^\infty  {\Big[ {1 - {{\Big( {\displaystyle\frac{1}{{1 + {x^{ - \alpha }}}}} \Big)}^r}} \Big]} xdx} \bigg)}}{{{{(1 + {\gamma _{th}})}^{rn}}}}.
\end{align}

\section{A Proof of  Lemma 3}
For $l\le 4$, the UE can not receive feedback, thus the number of interfering users remains unchanged in each repetition.
So we have
\begin{align}\label{trans_l4_proof}
    &{\Theta^{\rm Proa} [{n},{1},l]}=1-\prod\limits_{r = 1}^{l}(1-{\mathbb P}_{1,r}),
\end{align}
where
\begin{align}\label{m=1trans_proof}
       & {\mathbb P}_{1,r}={\eta_{1,r}}
        %\sum\limits_{n = 0}^\infty 
       {{\Theta^{\rm Proa} [{n},{1},1]}}\nonumber \\
      &=
\displaystyle\frac{{{\eta_{1,r}}\exp \Big(\displaystyle\frac{{ - {{\gamma_{\rm th}}}{\sigma ^2}}}{\rho }
%\displaystyle\frac{{ - k{\gamma _{th}}{\sigma ^2}}}{\rho } 
{ - \displaystyle{{{\eta_{1,r}}\frac{\lambda _a}{{{\lambda _B}}}}}\Big( {{}_2{F_1}\Big( { - \frac{2}{\alpha },1;\frac{{\alpha  - 2}}{\alpha }; - {{\gamma_{\rm th}}}} \Big) - 1} \Big)} \Big)}}{{{(1 + {{\gamma_{\rm th}}})}^{n}}}.
    \end{align}
For $l\ge 5$, the UE can receive feedback from the 4th repetition, thus the number of interfering users changes from the 5th
repetition.
So we have
\begin{align}\label{trans_l5_proof}
{\Theta^{\rm Proa} [{n},{1},l]} =1-(1-{{\Theta^{\rm Proa} [{n},{1},4]}})\Big(\prod\limits_{r = 5}^{l}{1-\mathbb P}_{1,r}\Big),
\end{align}
where
\begin{align}\label{m=1trans_proof2}
        {\mathbb P}_{1,r}&={\eta_{1,r}}{{\rm O}[{n},1,r]}
        %\sum\limits_{n = 0}^\infty 
       {{\Theta^{\rm Proa} [{n},{1},1]}},
%       \sum\limits_{{r} = 1}^{{l}} {{( - 1)}^{{r} + 1}}\Big( \begin{array}{l}
% {l}\\
% {r}
% \end{array} \Big)
% \exp \big( {\displaystyle{{ - k{\gamma _{th}}{\sigma ^2}}}/{\rho }\big)}
% \displaystyle\frac{{\exp \Big({{ - {\gamma _{th}}{\sigma ^2}}}/{\rho }
% %\displaystyle\frac{{ - k{\gamma _{th}}{\sigma ^2}}}{\rho } 
% { - \displaystyle{{{\eta_{1,r}}{\lambda _a}}}/{{{\lambda _B}}}\Big( {{}_2{F_1}\Big( { - \frac{2}{\alpha },1;\frac{{\alpha  - 2}}{\alpha }; - {\gamma _{th}}} \Big) - 1} \Big)} \Big)}}{{{(1 + {\gamma _{th}})}^{n}}}.
    \end{align}
    with
    \begin{align}\label{N_condition1mu_proa_proof}
{{\rm O}[{n},1,r]} = \displaystyle\frac{{{c^{(c + 1)}}\Gamma (n + c + 1){{\big( {{{{\eta_{1,r}\lambda _{a}}}}/{{{\lambda _B}}}} \big)}^{n}}}}{{\Gamma (c + 1)\Gamma (n+ 1){{\big( {{{{\eta_{1,r}\lambda _{a}}}}/{{{\lambda _B}}} + c} \big)}^{n + c + 1}}}}.
\end{align}

\bibliographystyle{IEEEtran}

\bibliography{IEEEabrv,grant}

\end{document}